\documentclass[a4paper,11pt]{article}

\usepackage{mathtools}
\usepackage{dsfont}
\usepackage{changepage}
\usepackage{algo}
\usepackage{algorithm, algorithmic, cite}
\usepackage{float,graphicx,verbatim,fullpage,hyperref,amssymb,amsmath,amsthm,amsfonts,enumerate,multicol, xspace,comment,xcolor,thm-restate,url, cite}
\usepackage[T1]{fontenc}
\usepackage[margin=2cm]{geometry}
\usepackage{amssymb,amsmath,amsthm,amsfonts}
\usepackage{subfig}
\usepackage{graphicx,comment,tikz}
\newcommand*{\email}[1]{\texttt{#1}}
\usepackage{hyperref}
\allowdisplaybreaks
\newcommand{\OPT}{\text{OPT}}
\newcommand{\mcs}{\mathcal{S}}
\newcommand{\mcp}{\mathcal{P}}
\newcommand{\mcq}{\mathcal{Q}}
\newcommand{\mcx}{\mathcal{X}}
\newcommand{\lpnormmc}{$\ell_p$\textsc{-norm-multiway-cut}\xspace}

\newcommand{\mwc}{\textsc{multiway-cut}\xspace}
\newcommand{\mmmwc}{\textsc{min-max-multiway-cut}\xspace}
\newcommand{\lpnormmwc}{$\ell_p$\textsc{-norm-multiway-cut}\xspace}
\newcommand{\UTC}{\text{UTC}}
\newcommand{\UTCalgo}{\text{UTC-BICRIT-ALGO}\xspace}
\newcommand{\mskp}{\textsc{min-sum-equi-$k$-partitioning}\xspace}

\newcommand{\R}{\mathbb{R}}
\newcommand{\E}{\mathbb{E}}

\newtheorem{theorem}{Theorem}[section]
\newtheorem{lemma}{Lemma}[section]

\newtheorem{proposition}{Proposition}[section]
\newtheorem{claim}{Claim}[section]

\def\final{0}  % set this to 1 to get a comment-free version
\def\iflong{\iffalse}
\ifnum\final=0  %namely if we allow comments in the output
\newcommand{\knote}[1]{{\color{red}[{\tiny Karthik: \bf #1}]\marginpar{\color{red}*}}}
\newcommand{\wnote}[1]{{\color{blue}[{\tiny Weihang: \bf #1}]\marginpar{\color{blue}*}}}
\newcommand{\todonote}[1]{{\color{red}[{\tiny TODO: \bf #1}]\marginpar{\color{red}*}}}
\else % in this case [final=1] we don't want any comments to show
\newcommand{\knote}[1]{}
\newcommand{\wnote}[1]{}
\newcommand{\todonote}[1]{}
\fi  %ok, we're done with defining comment macros
\def\set#1{\left\{ #1 \right\}}

\title{$\ell_p$-norm Multiway Cut\thanks{University of Illinois, Urbana-Champaign, Email: \email{\{karthe, weihang3\}@illinois.edu}. Supported in part by NSF grants CCF-1814613 and CCF-1907937.}
}

\author{
Karthekeyan Chandrasekaran
\and 
Weihang Wang%\footnotemark[1]
}
\date{}
\begin{document}

\maketitle

\begin{abstract}
    We introduce and study \lpnormmwc: the input here is an undirected graph with non-negative edge weights along with $k$ terminals and the goal is to find a partition of the vertex set into $k$ parts each containing exactly one terminal so as to minimize the $\ell_p$-norm of the cut values of the parts. This is a unified generalization of min-sum multiway cut (when $p=1$) and min-max multiway cut (when $p=\infty$), both of which are well-studied classic problems in the graph partitioning literature. 
    %Our motivation behind studying the $\ell_p$-norm objective is that it acts as a fairness inducing objective based on $p$. 
    We show that \lpnormmwc is NP-hard for constant number of terminals and is NP-hard in planar graphs. On the algorithmic side, we design an $O(\log^2 n)$-approximation for all $p\ge 1$. We also show an integrality gap of $\Omega(k^{1-1/p})$ for a natural convex program and an $O(k^{1-1/p-\epsilon})$-inapproximability for any constant $\epsilon>0$ assuming the small set expansion hypothesis. 
\end{abstract}

\newpage
\setcounter{page}{1}
\section{Introduction}\label{sec:intro}
\mwc is a fundamental problem in combinatorial optimization with both theoretical as well as practical motivations. 
%In this problem, we are given 
The input here is 
an undirected graph $G=(V,E)$ with non-negative edge weights $w:E\rightarrow \R_+$ along with $k$ specified terminals $T=\{t_1, t_2,\ldots, t_k\}\subseteq V$. The goal is to find a partition $\mcp=(P_1,P_2,\ldots,P_k)$ of the vertex set with $t_i\in P_i$ for each $i\in[k]$ so as to minimize the sum of the cut values of the parts, i.e., the objective is to minimize $\sum_{i=1}^k w(\delta(P_i))$, where $\delta(P_i)$ denotes the set of edges with exactly one end-vertex in $P_i$ and $w(\delta(P_i)):=\sum_{e\in \delta(P_i)}w(e)$. 
%given an undirected graph with non-negative edge weights along with $k$ specified terminal vertices. The goal is to partition the vertices into $k$ parts such that each part contains exactly one terminal so as to minimize the sum of the cut values of the parts. We recall that the cut value of a subset of vertices is the total weight of edges with exactly one end-vertex in the subset. 
On the practical side, \mwc has been used to model file-storage in networks as well as partitioning circuit elements among chips---see \cite{DJPSY94, ST04}. 
On the theoretical side, \mwc generalizes the min $(s,t)$-cut problem which is polynomial-time solvable. 
%for $k=2$ terminals corresponds to the fundamental min $(s,t)$-cut problem. 
In contrast to min $(s,t)$-cut, \mwc is NP-hard for $k\ge 3$ terminals \cite{DJPSY94}. The algorithmic study of \mwc has led to groundbreaking rounding techniques and integrality gap constructions in the field of approximation algorithms \cite{CKR00, CCT06, KKSTY04, BNS13, BSW17, SV14, MNRS08, AMM17, BCKM20} and novel graph structural techniques in the field of fixed-parameter algorithms \cite{Marx06}. It is known that \mwc does not admit a  $(1.20016-\epsilon)$-approximation for any constant $\epsilon>0$ assuming the Unique Games Conjecture \cite{BCKM20} and the currently best known approximation factor is $1.2965$ \cite{SV14}. 
%The current best approximation factor for \mwc is $1.2965$ \cite{??} and moreover, it does not admit a $(1.20016-\epsilon)$-approximation for any constant $\epsilon>0$ assuming the Unique Games Conjecture \cite{??}. 

Motivated by its connections to partitioning and clustering, Svitkina and Tardos \cite{ST04} introduced a local part-wise min-max objective for \mwc---we will denote this problem as \mmmwc: The input here is the same as \mwc while the goal is to find a partition  $\mcp=(P_1,P_2,\ldots,P_k)$ of the vertex set with $t_i\in P_i$ for each $i\in[k]$ so as to minimize  $\max_{i=1}^k w(\delta(S))$. We note that \mwc and \mmmwc differ only in the objective function---the objective function in \mwc is to minimize the sum of the cut values of the parts while the objective function in \mmmwc is to minimize the max of the cut values of the parts. 
%\mmmwc aims to ensure that no part pays too much in cut value, thus being fair to the parts in the partition.
\mmmwc can be viewed as a fairness inducing multiway cut as it aims to ensure that no part pays too much in cut value. 
Svitkina and Tardos showed that \mmmwc is NP-hard for $k\ge 4$ terminals and also that it admits an $O(\log^3{n})$-approximation. Bansal, Feige, Krauthgamer, Makarychev, Nagarajan, Naor, and Schwartz subsequently improved the approximation factor to $O(\sqrt{\log{n}\log{k}})$ (which is $O(\log{n})$) \cite{BFKMNNS14}. 

In this work, we study a unified generalization of \mwc and \mmmwc that we term as \lpnormmwc: In this problem, the input is the same as \mwc, i.e., we are given an undirected graph $G=(V,E)$ with non-negative edge weights $w:E\rightarrow \R_+$ along with $k$ specified terminal vertices $T=\{t_1, t_2,\ldots, t_k\}\subseteq V$. The goal is to find a partition 
$\mcp=(P_1,P_2,\ldots,P_k)$ of the vertex set with $t_i\in P_i$ for each $i\in[k]$ so as to minimize the $\ell_p$-norm of the cut values of the $k$ parts---formally, we would like to minimize  
\[\left(\sum_{i=1}^k \left(\sum_{e\in \delta(P_i)} w(e) \right)^p\right)^{\frac{1}{p}}.\]
Throughout, we will consider $p\ge 1$. We note that \lpnormmwc for $p=1$ corresponds to \mwc and for $p=\infty$ corresponds to \mmmwc. We emphasize that \lpnormmwc could also be viewed as a multiway cut that aims for a stronger notion of fairness than \mwc but a weaker notion of fairness than \mmmwc. 
For $k=2$ terminals, \lpnormmwc reduces to min $(s,t)$-cut for all $p\ge 1$ and hence, can be solved in polynomial time. 
%We present a systematic study of \lpnormmwc for $p>1$ that it differs substantially from the case of $p=1$, is NP-hard even for constant number of terminals, 

\subsection{Our Results}

%\noindent \textbf{\mwc vs \lpnormmwc.}
We begin by remarking that there is a fundamental structural difference between \mwc and \lpnormmwc for $p>1$ (i.e., between $p=1$ and $p>1$). The optimal partition to \mwc satisfies a nice structural property: assuming that the input graph is connected, every part in an optimal partition for \mwc will induce a connected subgraph. Consequently, \mwc is also phrased as the problem of deleting a least weight subset of edges so that the resulting graph contains $k$ connected components with exactly one terminal in each component. However, this nice structural property does not hold for \lpnormmwc for $p>1$ as illustrated by the example in Figure \ref{fig:disconnected}. We remark that Svitkina and Tardos made a similar observation suggesting that the nice structural property fails for \mmmwc, i.e., for $p=\infty$---in contrast, our example in Figure \ref{fig:disconnected} shows that the nice structural property fails to hold for every  $p>1$.  
\begin{figure}[H]
    \centering
    \includegraphics[width=0.7\textwidth]{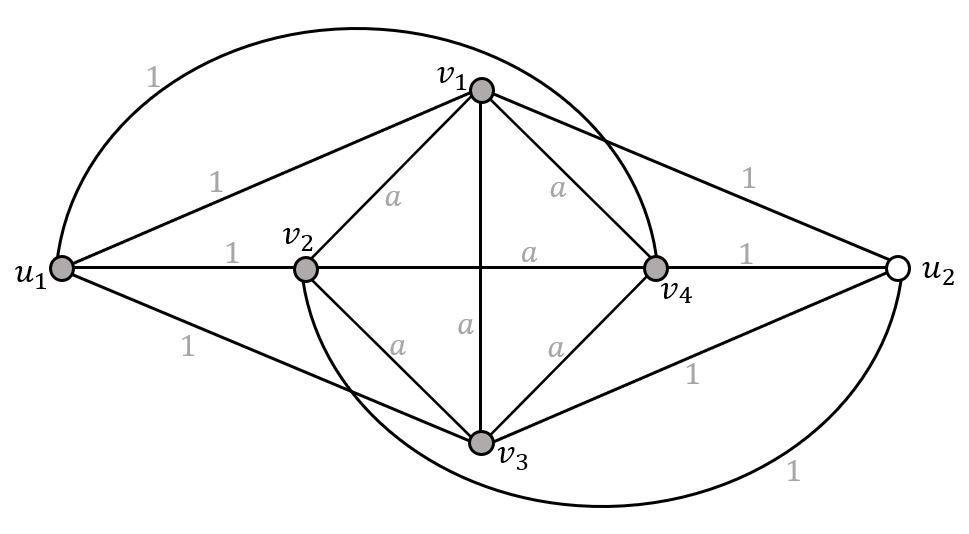}
    \caption{
    An example where the unique optimum partition for \lpnormmwc for $k=5$ induces a disconnected part for every $p>1$. The edge weights are as shown with $a:=8^{p/(p-1)}$ and the set of terminals is  $\{u_1,v_1,v_2,v_3,v_4\}$. A partition that puts $u_2$ with one of the terminals in $\{v_1, v_2, v_3, v_4\}$ (and isolates the remaining terminals) has $\ell_p$-norm objective value $((3a+3)^p + 3(3a+2)^p + 4^p)^{1/p}$ and the partition that puts $u_2$ with $u_1$ (and isolates the remaining terminals) has $\ell_p$-norm objective value $(4(3a+2)^p + 8^p)^{1/p}$---the latter is strictly cheaper by the choice of $a$. 
}
    \label{fig:disconnected}
\end{figure}
\iffalse
\begin{align*}
    ((3a+3)^p+3(3a+2)^p+4^p)-(4(3a+2)^p+8^p)&>(3a+3)^p-(3a+2)^p-8^p
    \\&\geq p(3a+2)^{p-1}-8^p>0.
\end{align*}
\fi

We now discuss our hardness results for \lpnormmwc. 
%We show the following hardness results: 
\begin{theorem}\label{thm:np-hardness}
We have the following hardness results for \lpnormmwc. 
\begin{enumerate}
    \item \lpnormmwc is NP-hard for every $p>1$ and every $k\ge 4$. 
    \item \lpnormmwc in planar graphs is NP-hard for every $p>1$. 
\end{enumerate}
\end{theorem}
We note that the case of $p=1$ and $p=\infty$ are already known to be hard: \mwc is NP-hard for $k=3$ terminals and is NP-hard in planar graphs when $k$ is arbitrary (i.e., when $k$ is not a fixed constant) \cite{DJPSY94}; \mmmwc is NP-hard for $k=4$ terminals and is NP-hard in trees when $k$ is arbitrary \cite{ST04}. Our NP-hardness in planar graphs result also requires $k$ to be arbitrary. 

%\mwc and \mmmwc are known to be NP-hard for constant number of terminals and even in planar graphs. 

%Moreover, 

Given that the problem is NP-hard, we focus on designing approximation algorithms. We show the following result: 

\begin{theorem}\label{thm:approx-algo}
There exists a polynomial-time $O(\log^{1.5}n\log^{0.5}{k})$-approximation for \lpnormmc for every $p\ge 1$, where $n$ is the number of vertices and $k$ is the number of terminals in the input instance. 
\end{theorem}

We note that our approximation factor is $O(\log^2 n)$ since $k\le n$. While it might be tempting to design an approximation algorithm by solving a convex programming relaxation for \lpnormmwc and rounding it, we rule out this approach: the natural convex programming relaxation has an integrality gap of $\Omega(k^{1-1/p})$---see Section \ref{sec:convex-program}. 
%This integrality gap is exhibited by the star graph. 
Hence, our approach for the approximation algorithm is not based on a convex program but instead based on combinatorial techniques. 
%Hence, our algorithm is based on combinatorial techniques. 

For comparison, we state the currently best known approximation factors for $p=1$ and $p=\infty$: \mwc admits a $1.2965$-approximation via an LP-based algorithm \cite{SV14} and \mmmwc admits an $O(\sqrt{\log{n}\log{k}})$-approximation based on a bicriteria approximation for the small-set expansion problem \cite{BFKMNNS14}. 
%Our approximation algorithm in Theorem \ref{thm:approx-algo} is inspired by the approximation algorithm of \cite{BFKMNNS14}. 

As a final result, we show that removing the dependence on the number $n$ of vertices in the approximation factor of \lpnormmwc is hard assuming the small set expansion hypothesis \cite{RST12}---see Section \ref{sec:inapprox}. In particular, we show that achieving a $(k^{1-1/p-\epsilon})$-approximation for any constant $\epsilon>0$ is hard. We note that there is a trivial  $O(k^{1-1/p})$-approximation for \lpnormmwc (see Section \ref{sec:trivial-approx}). 
\iffalse
For this, we consider \mskp: the input to this problem is a graph $G=(V,E)$ (where $n:=|V|$), an edge weight function $w: E\rightarrow \R_+$, and an integer $k\le n$. The goal is to partition $V$ into $k$ sets $P_1, \ldots, P_k$ such that $|P_i|=n/k$ for all $i\in [k]$ so as to minimize $\sum_{i=1}^k w(\delta(P_i))$. We will use $\lambda$ to denote the optimum objective value of \mskp. A partition $(P_1, \ldots, P_k)$ of $V$ is a $(\alpha, \beta)$-bicriteria approximation for \mskp if  $\sum_{i=1}^k w(\delta(P_i)) \le \alpha \lambda$ and $|P_i|\le \beta(n/k)$ for all $i\in [k]$. 
For constant $k$, it is known that $(O(1), O(1))$-bicriteria approximation is at least as hard as small set expansion \cite{??}. 
We show the following result which implies that a $k^{1-1/p-\epsilon}$-approximation is unlikely for \lpnormmwc:
\begin{theorem}\label{thm: mskp bic}
If \lpnormmc admits an efficient $k^{1-1/p-\epsilon}$-approximation algorithm for some constant $\epsilon>0$, then \mskp admits a $(O(k^{2-1/p}),O(1))$-bicriteria approximation for sufficiently large $k$.
\end{theorem}
We mention that our proof of Theorem \ref{thm: mskp bic} is similar to the result of Bansal et al \cite{BFKMNNS14} showing that \mmmwc does not admit a $k^{1-\epsilon}$-approximation assuming the smallset hypothesis. 
\fi

\subsection{Outline of techniques}
We briefly outline the techniques underlying our results. 
%Our results are inspired by the known hardness and algorithmic results for \mmmwc. 

\paragraph{Hardness results.} We show hardness of \lpnormmwc for $k=4$ terminals by a reduction from the graph bisection problem (see Section \ref{sec:np-hard-constant-many-terminals} for a description of this problem). Our main tool to control the $\ell_p$-norm objective in our hardness reduction is the Mean Value Theorem and its consequences (see Propositions \ref{prop:MVT} and \ref{prop:two mean values}). 
%We need the extra tool of mean value theorem to control this objective in the reduction. 
%The NP-hardness result for \mmmwc due to Svitkina and Tardos was also via a reduction from the graph bisection problem. Since we consider the $\ell_p$-norm objective, we need further ideas beyond that of Svitkina and Tardos' reduction. 
In order to show NP-hardness of \lpnormmwc in planar graphs, we reduce from the $3$-partition problem (see Section \ref{sec:np-hard-planar-graphs} for a description of this problem). We do a gadget based reduction where the gadget is planar. We note that the number of terminals in this reduction is not a constant and is $\Omega(n)$, where $n$ is the number of vertices. 
Once again, we rely on the Mean Value Theorem and its consequences to control the $\ell_p$-norm objective in the reduction. 
We mention that the starting problems in our hardness reductions are inspired by the hardness results shown by Svitkina and Tardos for \mmmwc: they showed that \mmmwc is NP-hard for $k=4$ terminals by a reduction from the graph bisection problem and that \mmmwc is NP-hard in trees by a reduction from the $3$-partition problem. We also use these same starting problems, but our reductions are more involved owing to the $\ell_p$-norm nature of the objective. 

\paragraph{Approximation algorithm.} For the purposes of the algorithm, we will assume knowledge of the optimum value, say $\OPT$---such a value can be guessed within a factor of $2$ via binary search. Our approximation algorithm proceeds in three steps. We describe these three steps now. 

In the first step of the algorithm, we obtain a collection $\mcs$ of subsets of the vertex set satisfying four properties: (1) each set $S$ in the collection $\mcs$ has at most one terminal, (2) the $\ell_p$-norm of the cut values of the sets in the collection raised to the $p$th power is small, i.e., $\sum_{S\in \mcs}w(\delta(S))^p=(\beta^p \log n) \OPT^p$ where $\beta = O(\sqrt{\log{n}\log{k}})$, (3) the number of sets in the collection $\mcs$ is  $O(k\log{n})$, and (4) the union of the sets in the collection $\mcs$ is $V$. We perform this first step via a multiplicative updates method. For this, we use a bicriteria approximation algorithm for the unbalanced terminal cut problem which was given by Bansal et al \cite{BFKMNNS14} (see Section \ref{sec:prelims} for a description of the unbalanced terminal cut problem and the bicriteria approximation). 

Although property (2) gives a bound on the $\ell_p$-norm of the cut values of the sets in the collection $\mcs$ relative to the optimum, the collection $\mcs$ does not correspond to a feasible multiway cut: recall that a feasible multiway cut is a  partition $\mcp=(P_1, \ldots, P_k)$ of the vertex set where each $P_i$ contains exactly one terminal. The objective of the next two steps is to refine the collection $\mcs$ to achieve feasibility without blowing up the $\ell_p$-norm of the cut values of the parts. 

In the second step of the algorithm, we uncross the sets in the collection $\mcs$ to obtain a partition $\mcq$ without increasing the cut values of the sets. We crucially exploit the posimodularity property of the graph cut function to achieve this: posimodularity states that for all subsets $A, B\subseteq V$ of vertices, either $w(\delta(A))\ge w(\delta(A-B))$ or $w(\delta(B))\ge w(\delta(B-A))$. We iteratively consider all pairs of crossing subsets $A,B$ in the collection $\mcs$ and replace $A$ with $A-B$ if $w(\delta(A))\ge w(\delta(A-B))$  or replace $B$ with $B-A$ if $w(\delta(B))\ge w(\delta(B-A))$. The outcome of this step is a partition $\mcq$ of the vertex set $V$ satisfying three properties: (i) each part $Q$ in the partition $\mcq$ has at most one terminal, (ii) the $\ell_p$-norm of the cut values of the parts in the partition $\mcq$ raised to the $p$th power is still small, i.e., $\sum_{Q\in \mcq}w(\delta(Q))^p=(\beta^p \log n) \OPT^p$, and (iii) the number of parts in the partition $\mcq$ is $O(k\log{n})$. 

Once again, we observe that the partition $\mcq$ at the end of the second step may not correspond to a feasible multiway cut: we could have more than $k$ parts in $\mcq$ with some of the parts having no terminals. We address this issue in the third step by a careful aggregation. 

For the third step of the algorithm, let $Q_i$ be the part in $\mcq$ that contains terminal $t_i$---we have $k$ such parts by property (i)---and let $R_1, \ldots, R_t$ be the remaining parts in $\mcq$ that contain no terminals. We will aggregate the remaining parts of $\mcq$ into the $k$ parts $Q_1, \ldots, Q_k$ without blowing up the $\ell_p$-norm of the cut value of the parts. By property (iii), the number of remaining parts $t$ is $O(k\log{n})$.  We create $k$ disjoint buckets $B_1, \ldots, B_k$ where $B_i$ contains the union of $O(\log{n})$ many parts among $R_1, \ldots, R_t$. Finally, we merge $B_i$ with $Q_i$. This results in a partition $\mcp=(Q_1\cup B_1, \ldots, Q_k\cup B_k)$ of $V$ with terminal $t_i$ being in the $i$th part $Q_i\cup B_i$. The key now is to control the blow-up in the $p$th power of the $\ell_p$-norm of the cut values of the parts in $\mcp$: we bound this by a $O(\log^{p-1}{n})$-factor relative to the $p$th power of the $\ell_p$-norm of the cut values of the parts in $\mcq$ via Jensen's inequality; while using Jensen's inequality, we exploit the fact that each bucket contained $O(\log{n})$ many parts. Consequently, using property (ii), the $\ell_p$-norm objective value of the cut values of the parts in the partition $\mcp$ raised to the $p$th power is still small---we show that $\sum_{P\in \mcp}w(\delta(P))^p=\beta^p \log^p n \OPT^p$ and hence, we have an approximation factor of $O(\beta \log{n})$. 

%We briefly discuss the individual steps now. We perform the first step via a multiplicative updates method. For this, we use a bicriteria approximation algorithm for the unbalanced terminal cut problem which was given by Bansal et al \cite{BFKMNNS14}. This 
The first step of our algorithm is inspired by the $O(\log{n})$-approximation algorithm for \mmmwc due to Bansal et al \cite{BFKMNNS14}---we modify the multiplicative weights update method and adapt it for \lpnormmwc. Our second and third steps differ from that of Bansal et al. We mention that the second and third steps of our algorithm can be adapted to achieve an $O(\beta \log{n})$-approximation factor for \lpnormmwc for $p=\infty$, but the resulting approximation factor is only $O(\log^2{n})$ which is weaker than the $O(\log{n})$-factor achieved by Bansal et al. The additional loss of $\log{n}$-factor in our algorithm comes from the third step (i.e., the aggregation step). The aggregation step designed in \cite{BFKMNNS14} is randomized and saves the $\log{n}$-factor in expectation, but it does not generalize to \lpnormmwc.
As mentioned before, the second step of our algorithm relies on posimodularity. The posimodularity property of the graph cut function has been used in previous works for \mmmwc in an implicit fashion by a careful and somewhat tedious edge counting argument \cite{ST04, BFKMNNS14}. We circumvent the edge counting argument here by the clean posimodularity abstraction. Moreover, the posimodularity abstraction makes the counting considerably easier for our more general problem of \lpnormmwc. 
%  We mention that the $O(\log{n})$-approximation for \mmmwc given by Bansal et al differs substantially from our second two steps. 

%We obtain such a collection by relying on a bicriteria approximation algorithm for 

\subsection{Related Work}
%As mentioned earlier, 
\lpnormmwc can be viewed as a fairness inducing objective in the context of multiway partitioning problems. Recent works have proposed and studied various fairness inducing objectives for graph cuts and partitioning that are different from \lpnormmwc. We briefly discuss these works here. All of the works mentioned in this subsection consider a more general problem known as \emph{correlation clustering}---we discuss these works by specializing to cut and partitioning problems since these specializations are the ones related to our work. 

Puleo and Milenkovic \cite{PM18} introduced a local vertex-wise min-max objective for min $(s,t)$-cut---here, the goal is to partition the vertex set $V$ of the given edge-weighted undirected graph into two parts $(S, V\setminus S)$ each containing exactly one of the terminals in $\{s,t\}$ so as to minimize  $\max_{v\in V}w(\delta(v)\cap \delta(S))$. The motivation behind this objective is that the cut should be fair to every vertex in the graph, i.e., no vertex should pay a lot for the edges in the cut.  A result of Chv\'{a}tal \cite{Chvatal84} implies that this problem is  $(2-\epsilon)$-inapproximable for every constant $\epsilon>0$. Charikar, Gupta, and Schwartz \cite{CGS17} gave an  $O(\sqrt{n})$-approximation for this problem. Reducing the approximability vs inapproximability gap for this problem remains an intriguing open problem. Kalhan, Makarychev, and Zhou \cite{KMZ19} considered an $\ell_p$-norm version of the objective where the goal is to minimize $(\sum_{v\in V} w(\delta(v)\cap \delta(S))^p)^{1/p}$ and gave an  $O(n^{\frac{1}{2}-\frac{1}{2p}}\log^{\frac{1}{2}-\frac{1}{2p}}{n})$-approximation, thus interpolating the best known results for $p=1$ and $p=\infty$. 

Ahmadi, Khuller, and Saha \cite{AKS19} introduced a min-max version of multicut: the input consists of an undirected graph $G=(V,E)$ with edge weights $w:E\rightarrow \R_+$ along with source-sink terminal pairs $(s_1, t_1), \ldots, (s_k, t_k)$. The goal is to find a  partition $\mcp=(P_1, \ldots, P_r)$ of the vertex set with all source-sink pairs separated by the partition so as to minimize $\max_{i\in [r]}w(\delta(P_i))$. We emphasize that the number of parts here---namely, $r$---is not constrained by the input and hence, could be arbitrary. Ahmadi, Khuller, and Saha gave an $O(\sqrt{\log{n}\max\{\log{k}, \log{T}\}})$-approximation for this problem, where $T$ is the number of parts in the optimal solution. Kalhan, Makarychev, and Zhou \cite{KMZ19} improved the approximation factor to $2+\epsilon$. 

\paragraph{Organization.} We begin with preliminaries in Section \ref{sec:prelims}. We present the complete details of our approximation algorithm and prove Theorem \ref{thm:approx-algo} in Section \ref{sec:approx-algo}. 
%Due to space limitations, the proofs of the rest of the results appear in the appendix. 
We show the hardness results and prove Theorem \ref{thm:np-hardness} in 
%Appendix 
Section 
\ref{sec:np-hardness}. We discuss a convex program and its integrality gap in 
%Appendix 
Section 
\ref{sec:convex-program}. We discuss the inapproximability and present a trivial $O(k^{1-1/p})$-approximation in 
%Appendix  
Section
\ref{sec:inapprox}. We conclude with a few open problems in Section \ref{sec:conclusion}.

\section{Preliminaries}\label{sec:prelims}

We start with notations that will be used throughout. Let $G=(V,E)$ be an undirected graph with edge weight function $w:E\to \R_+$ and vertex weight function $y:V\to \R_+$.
For every subset $S\subseteq V$, we use $\delta_G(S)$ to denote the set of edges that have exactly one end-vertex in $S$ (we will drop the subscript $G$ when the graph is clear from context), and we write $w(\delta(S)):=\sum_{e\in\delta(S)}w(e)$. Moreover, we will use $y(S)$ to refer to $\sum_{v\in S}y(v)$.
We will denote an instance of \lpnormmwc by $(G, w, T)$, where $G=(V,E)$ is the input graph, $w:E\rightarrow \R_+$ is the edge weight function, and $T\subseteq V$ is the set of terminal vertices. We will call a partition $\tilde{\mcp}=(P_1, \ldots, P_r)$ of the vertex set to be a multiway cut if $r=k$ and $t_i\in P_i$ for each $i\in [k]$ and denote the $\ell_p$-norm of the cut values of the parts (i.e., $(\sum_{i=1}^k w(\delta(P_i))^p)^{1/p}$) as the \emph{$\ell_p$-norm objective value of the multiway cut $\tilde{\mcp}$}. 

We note that the function $\mu:\R\to\R$ defined by $\mu(x):=x^p$ is convex for every $p\geq 1$.
We will use Jensen's inequality as stated below in our approximation algorithm as well as our hardness reductions. %order to bound the $\ell_p$-norm in our approximation factor. 
\begin{lemma}[Jensen]\label{lem:jensen}
Let $\mu:\R\to\R$ be a convex function. For arbitrary $x_1,\ldots,x_t\in\R$, we have 
\[\mu\left(\frac{1}{t}\sum_{i=1}^t x_i\right)\leq\frac{1}{t}\sum_{i=1}^t\mu(x_i).\]
\end{lemma}

Our algorithm relies on the graph cut function being symmetric and submodular. We recall that the graph cut function $f: 2^V\rightarrow \R_+$ is given by $f(S):=w(\delta(S))$ for all $S\subseteq V$. Let $f:2^V\rightarrow \R_+$ be a set function. The function $f$ is symmetric if $f(S) = f(V\setminus S)$ for all $S\subseteq V$, submodular if $f(A) + f(B)\ge f(A\cap B)+ f(A\cup B)$ for all $A,B\subseteq V$, and posimodular if $f(A) + f(B)\ge f(A-B) + f(B-A)$ for all $A,B\subseteq V$. Symmetric submodular functions are also posimodular (see Proposition \ref{prop:posimodular})---this fact has been used implicitly \cite{ST04, BFKMNNS14} and explicitly \cite{chekuri-ene-11, CC21} before. 

%The graph cut function $w\circ\delta$ is known to be both symmetric and submodular. The following proposition shows that $w\circ\delta$ is posimodular as well.
\begin{proposition}\label{prop:posimodular}
Symmetric submodular functions are posimodular.
\end{proposition}
\begin{proof}
Let $f:2^V\to\R$ be a symmetric submodular set function on a set $V$, and let $A,B\subseteq V$ be two arbitrary subsets. Then, we have
\begin{align*}
    f(A)+f(B)&=f(V-A)+f(B)\geq f((V-A)\cup B)+f((V- A)\cap B)
    \\&=f(V-(A-B))+f(B-A)=f(A-B)+f(B-A).
\end{align*}
In the above, the first and last equations follow by symmetry and the inequality follows by submodularity. 
\end{proof}

Our algorithm for \lpnormmc relies on an intermediate problem, namely the Unbalanced Terminal Cut problem that we introduce now. In Unbalanced Terminal Cut (UTC), the input  $(G,w,y,\tau, T)$ consists of an undirected graph $G=(V,E)$, an edge weight function $w:E\to\R_+$, a vertex weight function $y:V\to\R_+$, a real value $\tau\in[0,1]$, and a set $T\subseteq V$ of terminal vertices. The goal is to compute
\[\UTC(G,w,y,\tau, T):=\min\left\{w(\delta(S)):S\subseteq V, y(S)\geq \tau \cdot y(V),|S\cap T|\leq 1\right\}.\]
Bansal et al. gave a bicriteria approximation for UTC that is summarized in the theorem below.

\begin{theorem}\label{thm:UTC}\cite{BFKMNNS14}
There exists an algorithm \UTCalgo that takes as input $(G,w,y,\tau, T)$ consisting of an undirected graph $G=(V,E)$, an edge weight function $w:E\to\R_+$, a vertex weight function $y:V\to\R_+$, a number $\tau\in[0,1]$, and a set $T\subseteq V$ of terminal vertices and runs in polynomial time to return a set $S\subseteq V$ such that 
\begin{enumerate}
    \item $|S\cap T|\leq 1$,
    \item $y(S)=\Omega(\tau) y(V)$, and %\wnote{do I replace all $\Theta(\tau)$ with $d\tau$ for a constant $d$?}
    \item $w(\delta(S))\leq\alpha\UTC(G,w,y,\tau, T)$, where $\alpha=O(\sqrt{\log n\log(1/\tau)})$ and $n=|V|$.
\end{enumerate}
\end{theorem}

\section{Approximation Algorithm}
\label{sec:approx-algo}
\iffalse
\begin{theorem} \label{thm:main algo}
\wnote{Main theorem} There exists an algorithm that takes as input an undirected graph $G=(V,E)$, an edge weight function $w:E\to\R_+$ and $k$ distinct vertices $T:=\{t_1,\ldots,t_k\}$, runs in polynomial time, and outputs an $O(\log^2 n)$-approximate multiway cut $\mcq=(Q_1,\ldots,Q_k)$, i.e. 
\begin{align*}
    &\left(\sum_{i=1}^k w(\delta(Q_i))^p\right)^{1/p}\leq O(\log^2 n)\OPT,
\end{align*}
where 
\[\OPT:=\min\left\{\left(\sum_{i=1}^k w(\delta(P_i))^p\right)^{1/p}:(P_1,\ldots,P_k)\text{ is a partition of }V,t_i\in P_i\,\forall i\in[k]\right\}.\]
\end{theorem}
\fi

Let $\OPT$ be the optimal $\ell_p$-norm objective value of a multiway cut in the given instance. 
For the purposes of the algorithm, we will assume knowledge of a value $D$ such that $D\ge \OPT^p$---such a value can be guessed via binary search. %\wnote{Do we need to mention what happens when $D<\OPT$?} \knote{We do not mention this explicitly usually. The typical argument is that we run the approximation for each choice of $D$ in increasing multiples of $2$---we output the least cost solution among all that we have computed. You could include this at the end if you feel the need.}

Our approximation algorithm to prove Theorem \ref{thm:approx-algo} involves three steps. In the first step of the algorithm, we will obtain a collection $\mcs$ of $O(k\log{n})$ sets whose union is the vertex set $V$ such that each set in the collection has at most one terminal, the cut value of each set is not too large relative to $D$, and the $\ell_p$-norm of the cut values of the sets in the collection is within a polylog$(n)$ factor of $D$ (see Lemma \ref{lemma:multi weight update}). Although the collection $\mcs$ has low $\ell_p$-norm value relative to $D$, the collection $\mcs$ may not be a feasible multiway cut. 
%partition satisfying the multiway cut constraint (i.e., $\mcs$ could contain more than $k$ sets and some of the sets could also have non-empty intersection). 
In the second step of the algorithm, we uncross the sets in the collection $\mcs$ without increasing the $\ell_p$-norm of the cut values of the sets in the collection (see Lemma \ref{lemma:uncross}). 
%We crucially use posimodularity to do this uncrossing to ensure that the cut value of each uncrossed set does not increase. Such an uncrossing immediately ensures that the $\ell_p$-norm of the cut values of the sets in the collection also does not increase. 
After uncrossing, we obtain a partition, but we could have more than $k$ sets. We address this in our third step, where we aggregate parts to ensure that we obtain exactly $k$ parts (see Lemma \ref{lemma:aggregate}). We rely on Jensen's inequality to ensure that the aggregation does not blow-up the $\ell_p$-norm of the cut values of the sets in the partition. 

We begin with the first step of the algorithm in Lemma \ref{lemma:multi weight update}.  
%To prove Theorem \ref{thm:main algo}, we will use a multistage algorithm to be described in Lemma \ref{lemma:multi weight update}, \ref{lemma:uncross} and \ref{lemma:aggregate}. For the purpose of designing algorithm, we assume knowledge of a real value $D$ such that $D\geq\OPT^p$ via binary search.

\begin{lemma}\label{lemma:multi weight update}
There exists an algorithm that takes as input an undirected graph $G=(V,E)$, an edge weight function $w:E\to\R_+$, $k$ distinct terminal vertices $T:=\{t_1,\ldots,t_k\}\subseteq V$ and a value $D>0$ such that there exists a partition $(P_1^*,\ldots,P_k^*)$ of $V$ with $t_i\in P_i^*$ for all $i\in[k]$ and $\sum_{i=1}^k w(\delta(P_i^*))^p\leq D$, and runs in polynomial time to return a collection of sets $\mcs\subseteq 2^V$ that satisfies the following:
\begin{enumerate}
    \item $|S\cap T|\leq 1$ and $w(\delta(S))\leq\beta(2D)^{1/p}$ for every $S\in \mcs$,
    \item $\sum_{S\in\mcs}w(\delta(S))^p=\beta^p(\log n) D$, and
    \item $|\mcs|=O(k\log n)$ and $|\{S\in\mcs:v\in S\}|\geq \log n$ for each $v\in V$,
\end{enumerate}
where $\beta=O(\sqrt{\log n\log k})$.
\end{lemma}

\begin{proof}
We will use Algorithm \ref{algo:multiweightsupdate} to obtain the desired collection $\mcs$. We will show the correctness of Algorithm \ref{algo:multiweightsupdate} based on Claims \ref{claim:existence of i_0}, \ref{claim:coverage} and \ref{claim:Lp norm of mcs}.

\begin{algorithm}
\caption{Multiplicative weights update}
\begin{algorithmic}
\STATE Initialize $t\leftarrow 1$, $\mcs\leftarrow \emptyset$, $y^1(v)=1$ for each $v\in V$, $Y^1= \sum_{v\in V}y^1(v)$ and $\beta=O(\sqrt{\log n\log k})$
\WHILE{$Y^t>\frac{1}{n}$}
\FOR{$i=1,2,\ldots,\log(2k)$}
\STATE Execute $\UTCalgo(G,w,y^t,2^{-i},T)$ to obtain a subset $S^t(i)\subseteq V$
\IF{$w(\delta(S^t(i)))\leq\beta(\frac{4D}{2^i})^{1/p}$}
\STATE Set $S^t=S^t(i)$ and BREAK
\ENDIF
\ENDFOR
\STATE $\mcs\gets\mcs\cup\{S^t\}$.
\FOR{$v\in V$}
\STATE Set $y^{t+1}(v)= \begin{cases}y^t(v)/2& \text{ if }v\in S^t,\\
y^t(v)&\text{ if } v\in V\setminus S^t.\end{cases}$
\ENDFOR
\STATE Set $Y^{t+1}=\sum_{v\in V}y^{t+1}(v)$.
\STATE Set $t\gets t+1$.
\ENDWHILE
\STATE Return $\mcs$
\end{algorithmic}
\label{algo:multiweightsupdate}
\end{algorithm}

%We now show correctness of Algorithm \ref{algo:multiweightsupdate}. 
Our first claim will help in showing that the set $S^t$ added in each iteration of the while loop satisfies certain nice properties.

\begin{claim}\label{claim:existence of i_0}
For every iteration $t$ of the while loop of Algorithm \ref{algo:multiweightsupdate}, there exists $i\in\{1,2,\ldots,\log(2k)\}$ such that the set $S^t(i)$ satisfies the following conditions:
\begin{enumerate}
    \item $|S^t(i)\cap T|\leq 1$,
    \item $y^t(S^t(i))=\Omega(\frac{Y^t}{2^i})$, and
    \item $w(\delta(S^t(i)))\leq\beta(\frac{4D}{2^i})^{1/p}$.
\end{enumerate}
\end{claim}
\begin{proof}
%Let $(S^\ast_1,S^\ast_2,\ldots,S^\ast_k)$ be an optimal \lpnormmc such that $t_i\in S^\ast_i$ for each $i\in[k]$. Then 
We have that $\sum_{i=1}^k y^t(P^\ast_i)=y^t(V)$ and 
\[\sum_{i=1}^k w(\delta(P^\ast_i))^p\leq D.\]
%by optimality of $(S^\ast_1,\ldots,S^\ast_k)$. 
%Moreover, we have $\sum_{i=1}^k y^t(P^\ast_i)=y^t(V)$.
%We will define an index set $L\subseteq [k]$ by
Let $L$ be the subset of indices of parts for which the cut value is relatively low: 
\[L:=\left\{j\in[k]:w(\delta(P^\ast_j))^p\leq \frac{2y^t(P^\ast_j)}{Y^t}\cdot D\right\}.\]
It follows that
\begin{align*}
    \sum_{j\in [k]\setminus L}y^t(P^\ast_j)<\sum_{j\in [k]\setminus  L}\frac{w(\delta(P^\ast_j))^p Y^t}{2D}\leq \frac{Y^t}{2}
\end{align*}
and hence, 
\[\sum_{j\in L}y^t(P^\ast_j)=Y^t-\sum_{j\in[k]\setminus  L}y^t(P^\ast_j)>Y^t-\frac{Y^t}{2}=\frac{Y^t}{2}.\]
Since $|L|\leq k$, there exists an index $q\in L$ such that $y^t(P^\ast_q)>Y^t/(2k)$. Let us fix $i_0$ to be an integer such that $y^t(P^\ast_q)\in (Y^t\cdot 2^{-i_0},Y^t\cdot 2^{-i_0+1}]$. Then, we must have $i_0\leq \log(2k)$. We note that the set $P^\ast_q$ satisfies $|P^\ast_q\cap T|=1$ and $y^t(P^\ast_q)>Y^t/(2k)=y^t(V)/(2k)$. This implies $P^\ast_q$ is feasible to the UTC problem on input $(G,w,y^t,1/2^{i_0},T)$. Therefore, according to Theorem \ref{thm:UTC}, the set $S^t(i_0)$ has the following properties:  Firstly, $|S^t(i_0)\cap T|\leq 1$. Secondly, $y^t(S^t(i_0))= \Omega(1/2^{i_0})y^t(V)=\Omega(Y^t/2^{i_0})$. Finally,
\begin{align*}
    w(\delta(S^t(i_0)))&= O(\sqrt{\log n\log(2k)})\cdot\UTC\left(G,w,y^t,\frac{1}{2^{i_0}},T\right)
    \\&= O(\sqrt{\log n\log k})\cdot w(\delta(P^\ast_q))
    \\&= O(\sqrt{\log n\log k})\cdot\left(\frac{2y^t(P^\ast_q)}{Y^t}\cdot D\right)^{\frac{1}{p}}
    \\&= O(\sqrt{\log n\log k})\cdot\left(\frac{2\cdot Y^t\cdot 2^{-i_0+1}}{Y^t}\cdot D\right)^{\frac{1}{p}}
    \\&=O(\sqrt{\log n\log k})\cdot\left(\frac{4D}{2^{i_0}}\right)^{\frac{1}{p}}.
\end{align*}
This completes the proof of Claim \ref{claim:existence of i_0}. 
\end{proof}

For the rest of the proof, we will use the following notation: In the $t$'th iteration of the while loop of Algorithm \ref{algo:multiweightsupdate}, we will fix $i_t\in\{1,2,\ldots,\log(2k)\}$ to be the integer such that $S^t=S^t(i_t)$. We will use $\ell$ to denote the total number of iterations of the while loop. For each $v\in V$, We define $N_v:=|\{t\in[\ell]:v\in S^t\}|$ to be the number of sets in the collection $\mcs$ that contain the vertex $v$.

We observe that for each $v\in V$, we have $y^{\ell+1}(v)=2^{-N_v}$. Claim \ref{claim:existence of i_0} and Theorem \ref{thm:UTC} together imply that the $t$'th iteration of the while loop leads to a set $S^t$ being added to the collection $\mcs$ such that 
\begin{enumerate}
    \item $|S^t \cap T|\le 1$, 
    \item $y^t(S^t)=\Omega(\frac{Y^t}{2^{i_t}})$, and 
    \item $w(\delta(S^t))\le \beta(\frac{4D}{2^{i_t}})^{1/p}$.
\end{enumerate}

Our next claim shows that the number of iterations of the while loop executed in Algorithm \ref{algo:multiweightsupdate} is small. Moreover, the union of the sets in the collection $\mcs$ is the vertex set $V$. 
%each vertex is present in a sufficiently large number of sets in the collection $\mcs$. 

\begin{claim} \label{claim:coverage}
The number of iterations $\ell$ of the while loop satisfies
$\ell=O(k\log n)$. Moreover, $N_v\geq\log n$ for each $v\in V$.
\end{claim}
\begin{proof}
Upon termination of Algorithm \ref{algo:multiweightsupdate}, we must have $Y^{\ell+1}\leq 1/n$. Combining with the earlier observation that $y^{\ell+1}(v)=2^{-N_v}$ for every $v\in V$, we have that 
\[2^{-N_v}=y^{\ell+1}(v)\leq Y^{\ell+1}\leq \frac{1}{n},\]
which implies that $N_v\geq \log n$ for every $v\in V$. 

It remains to show that $\ell=O(k\log n)$. Consider the $t$th iteration of the while loop for an arbitrary $t\in[\ell]$. By property 2 
of the set $S^t$  stated above, we have that 
%conclusion 2 in  Claim \ref{claim:existence of i_0},
$y^t(S^t)\ge cY^t/2^{i_t}\geq cY^t/(2k)$ for some constant $c>0$. Consequently, 
\begin{align*}
    Y^{t+1}=Y^t-\frac{y^t(S^t)}{2}\leq Y^t-\frac{cY^t}{4k}=\left(1-\frac{c}{4k}\right)Y^t.
\end{align*}
Due to the termination condition of the while loop, we know that $Y^\ell>1/n$. Hence, 
\begin{align*}
    \frac{1}{n}<Y^\ell\leq \left(1-\frac{c}{4k}\right)^{\ell-1}Y^1=\left(1-\frac{c}{4k}\right)^{\ell-1}n\leq \exp\left(-\frac{c(\ell-1)}{4k}\right)n.
\end{align*}
Therefore, $\frac{c(\ell-1)}{4k}=O(\log n)$ which implies that $\ell=O(k\log n)$. This completes the proof of Claim \ref{claim:coverage}.
\end{proof}

The next claim bounds the $\ell_p$-norm of the cut values of the sets in the collection $\mcs$. 

\begin{claim}\label{claim:Lp norm of mcs}
The collection $\mcs$ returned by Algorithm \ref{algo:multiweightsupdate} satisfies
$\sum_{S\in\mcs}w(\delta(S))^p=O(\beta^p\log n)\cdot D$.
\end{claim}
\begin{proof}
Consider the $t$th iteration of the while loop for an arbitrary $t\in[\ell]$. By property 3 of the set $S^t$ stated above, 
%conclusion 3 in Claim \ref{claim:existence of i_0}, 
we have that $w(\delta(S^t))\leq \beta(4D/2^{i_t})^{1/p}$ and consequently, $2^{i_t} \leq 4D\beta^p\cdot w(\delta(S^t))^{-p}$. Moreover, by property 2 of the set $S^t$ stated above, 
%conclusion 2 in Claim \ref{claim:existence of i_0}, 
we have that $y^t(S^t)\geq cY^t/2^{i_t}$ for some constant $c>0$. Hence, 
\[y^t(S^t)\geq \frac{cY^t}{2^{i_t}}\geq \frac{cY^t\cdot w(\delta(S^t))^p}{\beta^p\cdot 4D}.\]
Therefore, 
\[Y^{t+1}=Y^t-\frac{y^t(S^t)}{2}\leq \left(1-\frac{c\cdot w(\delta(S^t))^p}{\beta^p\cdot 8D}\right)Y^t.\]

Using the fact that $Y^\ell>1/n$, we observe that
\begin{align*}
    \frac{1}{n}&<Y^\ell\leq Y^1\cdot\prod_{t=1}^{\ell-1}\left(1-\frac{c\cdot w(\delta(S^t))^p}{\beta^p\cdot 8D}\right)=n\cdot\prod_{t=1}^{\ell-1}\left(1-\frac{c\cdot w(\delta(S^t))^p}{\beta^p\cdot 8D}\right)
    \\&\leq n\cdot \prod_{i=1}^{\ell-1}\exp\left(-\frac{c\cdot w(\delta(S^t))^p}{\beta^p\cdot 8D}\right)=n\cdot\exp\left(-\frac{c\cdot \sum_{i=1}^{\ell-1}w(\delta(S^t))^p}{\beta^p\cdot 8D}\right).
\end{align*}
This implies that $\frac{c\cdot \sum_{i=1}^{\ell-1}w(\delta(S^t))^p}{\beta^p\cdot 8D}=O(\log n)$, and hence $\sum_{i=1}^{\ell-1}w(\delta(S^t))^p=O(\beta^p\log n)\cdot D$.

In the $\ell$'th iteration of the while loop, we have $w(\delta(S^\ell))\leq\beta(4D/2^{i_\ell})^{1/p}$ by 
property 3 of the set $S^t$ stated above 
%conclusion 3 of Claim \ref{claim:existence of i_0}, 
and hence $w(\delta(S^\ell))^p\leq\beta^p\cdot 4D/2^{i_\ell}\leq O(\beta^p D)$. Consequently, $\sum_{i=1}^{\ell}w(\delta(S^t))^p=O(\beta^p\log n)\cdot D$. This completes the proof of Claim \ref{claim:Lp norm of mcs}.
\end{proof}

We now show correctness of our algorithm to complete the proof of Lemma \ref{lemma:multi weight update}. Firstly, we note that every $S\in\mcs$ satisfies $|S\cap T|\leq 1$ by property 1 of the set $S^t$ stated above. 
%due to the sub-routine in Theorem \ref{thm:UTC}. 
Moreover, we have $w(\delta(S))\leq \beta(4D/2^i)^{1/p}\leq \beta(2D)^{1/p}$, which implies conclusion 1 in Lemma \ref{lemma:multi weight update}. Secondly, Conclusion 2 in Lemma \ref{lemma:multi weight update} is implied by Claim \ref{claim:Lp norm of mcs}. Finally, conclusion 3 of Lemma \ref{lemma:multi weight update} is implied by Claim \ref{claim:coverage} because each iteration of the while loop adds exactly one new set to the collection $\mcs$.

We now bound the run time of Algorithm \ref{algo:multiweightsupdate}. Each iteration of the while loop takes polynomial time due to Theorem \ref{thm:UTC}, and the number of iterations of the while loop is $O(k\log n)$. This implies that the total run time of Algorithm \ref{algo:multiweightsupdate} is indeed polynomial in the size of the input.
\end{proof}

The collection $\mcs$ that we obtain in Lemma \ref{lemma:multi weight update} may not be a partition. 
%that satisfies the multiway cut constraint (i.e., we need exactly one terminal in each part). 
Our next lemma will uncross the collection $\mcs$ obtained from Lemma \ref{lemma:multi weight update} to obtain a partition without increasing the cut values of the sets. 
%However, the partition returned by this theorem may not satisfy the multiway cut constraint (i.e., we could have more than $k$ parts with some of the parts not containing any terminals). Subsequently, we will aggregate the partition to obtain a $k$-partition satisfying the multiway cut constraint without losing too much on the $\ell_p$-norm of the cut value. 

\begin{lemma}\label{lemma:uncross}
There exists an algorithm that takes as input a collection $\mcs\subseteq 2^V$ of subsets of vertices satisfying the conclusions in Lemma \ref{lemma:multi weight update} and runs in polynomial time to return a partition $\tilde{\mcq}$ of $V$ such that
\begin{enumerate}
    \item $|Q\cap T|\leq 1$ for each $Q\in\tilde{\mcq}$,
    \item $\sum_{Q\in\tilde{\mcq}}w(\delta(Q))^p\leq \sum_{S\in\mcs}w(\delta(S))^p$, and
    \item the number of parts in $\tilde{\mcq}$ is $O(k\log{n})$.
\end{enumerate}
\end{lemma}

\begin{proof}
For convenience, we will define $f:2^V\to\R_+$ by $f(S):=w(\delta(S))$ for all $S\subseteq V$. We will use Algorithm \ref{algo:uncrossing} to obtain the desired partition $\tilde{\mcq}$ of $V$.
\iffalse
such that
\begin{enumerate}
    \item $|P\cap T|\leq 1$ for each $P\in\tilde{\mcp}$, and
    \item $\sum_{P\in\tilde{\mcp}}f(P)^p\leq \sum_{S\in\mcs}f(S)^p$.
\end{enumerate}
\fi 
\begin{algorithm}
\caption{Uncrossing}
\begin{algorithmic}
\STATE Initialize $\tilde{\mcq}\gets\mcs$
\WHILE{there exist distinct sets $A,B\in\tilde{\mcq}$ such that $A\cap B\neq\emptyset$}
\IF{$f(A)\geq f(A-B)$}
\STATE Set $A\gets A-B$
\ELSE
\STATE Set $B\gets B-A$
\ENDIF
\ENDWHILE
\STATE Return $\tilde{\mcq}$
\end{algorithmic}
\label{algo:uncrossing}
\end{algorithm}

We now prove the correctness of Algorithm \ref{algo:uncrossing}. We begin by observing that Algorithm \ref{algo:uncrossing} indeed outputs a partition of the vertex set: Firstly, the while loop enforces that the output $\tilde{\mcq}$ satisfies $A\cap B=\emptyset$ for all distinct $A,B\in\tilde{\mcq}$. Secondly, during each iteration of the while loop, the set $\bigcup_{Q\in\tilde{\mcq}}Q$ remains unchanged: In the iteration of the while loop that uncrosses $A,B\in\tilde{\mcq}$, let $A'$ and $B'$ denote the updated sets at the end of the while loop, respectively. Then we must have $A'\cup B'=(A-B)\cup B=A\cup B$ or $A'\cup B'=A\cup (B-A)=A\cup B$. In either case, since $A'\cup B'=A\cup B$, the set $\bigcup_{Q\in\tilde{\mcq}}Q$ remains unchanged after the update. Therefore, we have $\bigcup_{Q\in\tilde{\mcq}}Q=\bigcup_{S\in\mcs}S$. We recall that $\bigcup_{S\in\mcs}S=V$ by conclusion 3 of Lemma \ref{lemma:multi weight update}. Hence, $\tilde{\mcq}$ is indeed a partition of $V$.

Furthermore, each set $Q$ in the output $\tilde{\mcq}$ is a subset of some set $S\in\mcs$. This implies $|Q\cap T|\leq |S\cap T|\leq 1$, thus proving the first conclusion.

%To prove the second conclusion, we note that each iteration of the while loop does not increase $\sum_{P\in \tilde{\mcp}}f(P)^p$
To prove the second conclusion, we use posimodularity of $f$ as shown in Proposition \ref{prop:posimodular}. Namely, for every $A,B\subseteq V$, 
\[f(A)+f(B)\geq f(A-B)+f(B-A).\]
Therefore, at least one of the following two hold: either $f(A)\ge f(A-B)$ or $f(B)\ge f(B-A)$. This implies that, by the choice of the algorithm, $\sum_{Q\in \tilde{\mcq}}f(Q)^p$ does not increase. 

To see the third conclusion, we note that after each iteration of the while loop, the size of $\tilde{\mcq}$ is unchanged. Therefore, at the end Algorithm \ref{algo:uncrossing}, we have $|\tilde{\mcq}|=|\mcs|=O(k\log n)$ by Lemma \ref{lemma:multi weight update}.

Finally, we bound the run time as follows. At initialization, there are $O((k\log n)^2)$ pairs $(A,B)\in\tilde{\mcq}^2$ such that $A\cap B\neq\emptyset$. After each iteration of the while loop, the number of such pairs decreases by at least $1$. Therefore, the total number of iterations of the while loop is $O((k\log n)^2)$. Hence, Algorithm \ref{algo:uncrossing} indeed runs in polynomial time.

\iffalse
Therefore, in the case $f(P)<f(P-Q)$, we must have $f(Q)>f(Q-P)$. In the while loop that uncrosses $P,Q\in\tilde{\mcp}$, we will again use $P'$ and $Q'$ to denote the updated versions of $P$ and $Q$ at the end of the while loop, respectively. If $f(P)\geq f(P-Q)$, then 
\[f(P')^p+f(Q')^p=f(P-Q)^p+f(Q)^p\leq f(P)^p+f(Q)^p.\]
Otherwise, we have $f(Q)>f(Q-P)$, and hence
\[f(P')^p+f(Q')^p=f(P)^p+f(Q-P)^p<f(P)^p+f(Q)^p.\]
In either case, the value of $\sum_{P\in\tilde{\mcp}}f(P)^p$ is non-increasing in each while loop. This proves condition 2.

\knote{Argue the third conclusion. Also, argue the run-time.}
\fi
\end{proof}

The partition $\tilde{\mcq}$ that we obtain in Lemma \ref{lemma:uncross} may contain more than $k$ parts and hence, some of the parts may not contain any terminals. Our next lemma will aggregate the parts in $\tilde{\mcq}$ from Lemma \ref{lemma:uncross} to obtain a $k$-partition that contains exactly one terminal in each part while controlling the increase in the $\ell_p$-norm of the cut value of the parts. 

\begin{lemma}\label{lemma:aggregate}
There exists an algorithm that takes as input a partition $\tilde{\mcq}$ of $V$ satisfying the conclusions in Lemma \ref{lemma:uncross} and runs in polynomial time to return a partition $(P_1,P_2,\ldots,P_k)$ of $V$ such that
\begin{enumerate}
    \item $t_i\in P_i$ for each $i\in[k]$, and
    \item $\sum_{i=1}^k w(\delta(P_i))^p= O((\beta\log n)^p)\cdot D$.
\end{enumerate}
\end{lemma}
\begin{proof}
We will use Algorithm \ref{algo:aggregrating} on input $\tilde{\mcp}$ to obtain the desired partition. 

\begin{algorithm}
\caption{Aggregating}
\begin{algorithmic}
\STATE Let $\mathcal{F}=\{Q\in\tilde{\mcq}:Q\cap T=\emptyset\}$.
\STATE Let $\mcp'=\{Q\in\tilde{\mcq}:Q\cap T\neq\emptyset\}=\{Q'_1,\ldots,Q'_k\}$, where $t_i\in Q'_i$ for each $i\in[k]$.
\STATE Partition the sets in $\mathcal{F}$ into $k$ buckets $B_1,\ldots,B_k$ such that $|B_i|=O(\log n)$ for each $i\in[k]$ (arbitrarily).
\FOR{$i=1,2,\ldots,k$}
\STATE Set $P_i\gets Q'_i\cup\left(\bigcup_{A\in B_i}A\right)$
\ENDFOR
\STATE Return $(P_1,\ldots,P_k)$.
\end{algorithmic}
\label{algo:aggregrating}
\end{algorithm}

The run time of Algorithm \ref{algo:aggregrating} is linear in its input size. We now argue the correctness. 
We note that the third step in Algorithm \ref{algo:aggregrating} is possible because $|\mathcal{F}|\le|\tilde{\mcq}|=O(k\log n)$.

%By properties of $\tilde{\mcp}$, it follows that 
Since $|Q\cap T|\le 1$ for each $Q\in \tilde{\mcq}$, the tuple $(P_1,\ldots,P_k)$ returned by Algorithm \ref{algo:aggregrating} is indeed a partition of $V$ satisfying $t_i\in P_i$ for all $i\in [k]$. We will now bound $\sum_{i=1}^k f(P_i)^p$, where $f: 2^V\rightarrow \R_+$ is given by $f(S):=w(\delta(S))$ for all $S\subseteq V$. We have that 
\begin{align*}
    \sum_{i=1}^k f(P_i)^p&=\sum_{i=1}^k f\left(Q'_i\cup\left(\bigcup_{A\in B_i}A\right)\right)^p\leq \sum_{i=1}^k\left(f(Q'_i)+\sum_{A\in B_i}f(A)\right)^p.
\end{align*}
Since the number of sets in $B_i$ is $O(\log{n})$, we have the following using Jensen's inequality (Lemma \ref{lem:jensen}) for each $i\in [k]$:
\begin{align*}
    \left(f(Q'_i)+\sum_{A\in B_i}f(A)\right)^p \le 
    (|B_i|+1)^{p-1}\left(f(Q'_i)^p+\sum_{A\in B_i}f(A)^p \right)
    = O(\log^{p-1}{n})\left(f(Q'_i)^p+\sum_{A\in B_i}f(A)^p \right).
\end{align*}
%\wnote{Here $|B_i|$ should be $|B_i|+1$ due to the extra term. Also in Algo 3, we assumed $|B_i|=O(\log n)$ instead of $|B_i|\leq\log n-1$.}

Hence, 
\begin{align*}
    \sum_{i=1}^k f(P_i)^p
    &= \sum_{i=1}^k O(\log^{p-1}{n}) \left(f(Q'_i)^p+\sum_{A\in B_i}f(A)^p\right)
    = O(\log^{p-1}{n}) \sum_{Q\in\tilde{\mcq}}f(Q)^p
    \\&= O(\log^{p-1}{n}) \sum_{S\in\mcs}f(S)^p=\beta^p O(\log^{p}{n})D.
\end{align*}
%where inequality $(\ast)$ follows from Jensen's inequality. 
The last but one equality above is due to conclusion 2 of Lemma \ref{lemma:uncross}bbb and the last equality is due to conclusion 2 of Lemma \ref{lemma:multi weight update}. 
Hence, $\sum_{i=1}^k w(\delta(P_i))^p=\sum_{i=1}^k f(P_i)^p= O((\beta\log n)^p)D$.
\end{proof}

Lemmas \ref{lemma:multi weight update}, \ref{lemma:uncross}, and \ref{lemma:aggregate} together lead to an algorithm that takes as input an undirected graph $G=(V,E)$, an edge weight function $w:E\to\R_+$, $k$ distinct terminal vertices $T:=\{t_1,\ldots,t_k\}\subseteq V$, and a value $D>0$ such that there exists a partition $(P_1^*,\ldots,P_k^*)$ of $V$ with $t_i\in P_i^*$ for all $i\in[k]$ such that $\sum_{i=1}^k w(\delta(P_i^*))^p\leq D$, and runs in polynomial time to return a multiway cut $\mcp=(P_1,\ldots,P_k)$ such that 
\begin{align*}
    \left(\sum_{i=1}^kw(\delta(P_i))^p\right)^{\frac{1}{p}}=\left(O((\beta\log n)^p)D\right)^{\frac{1}{p}}=O(\beta\log n)D^{\frac{1}{p}}=O(\log^{1.5}n\log^{0.5}k)D^{\frac{1}{p}}.
\end{align*}

In order to prove Theorem \ref{thm:approx-algo}, we may use binary search to guess $D\in[\OPT^p,(2\OPT)^p]$ and run the above algorithm to obtain a multiway cut $\mcp=(P_1,\ldots,P_k)$ such that 
\[\left(\sum_{i=1}^kw(\delta(P_i))^p\right)^{\frac{1}{p}}=O(\log^{1.5}n\log^{0.5}k)D^{\frac{1}{p}}=O(\log^{1.5}n\log^{0.5}k)\OPT.\]
This completes the proof of Theorem \ref{thm:approx-algo}.

%\wnote{We note that the aggregation algorithm designed by \cite{BFKMNNS14} saves an extra $\log n$ factor in expectation for \mmmwc. Their aggregation algorithm works for only \mmmwc, and does not generalize to \lpnormmc.}

\section{NP-hardness}
\label{sec:np-hardness}
In this section, we show NP-hardness results for \lpnormmwc thereby proving Theorem \ref{thm:np-hardness}. In Section \ref{sec:np-hard-constant-many-terminals}, we show that \lpnormmwc is NP-hard for $k=4$ terminals for every $p>1$ by a reduction from graph bisection. In Section \ref{sec:np-hard-planar-graphs}, we show that \lpnormmwc is NP-hard in planar graphs for every $p>1$ by a reduction from $3$-partition. 
In our reductions, we will frequently use the following two consequences of the Mean Value Theorem. We recall that the function $\mu(x) = x^p$ is differentiable. 
\begin{proposition}\label{prop:MVT}
For a differentiable function $\mu:\R\to\R$, and two real values $x\leq y$, we have
\[(y-x)\min_{z\in[x,y]}\mu'(z)\leq\mu(x)-\mu(y)\leq (y-x)\max_{z\in[x,y]}\mu'(z).\]

% \knote{Shouldn't it be 
% \[(y-x)\min_{z\in[x,y]}\mu'(z)\leq\mu(y)-\mu(x)\leq (y-x)\max_{z\in[x,y]}\mu'(z)?\]
% }
\end{proposition}

\begin{proposition}\label{prop:two mean values}
For $p\geq 1$ and real values $0<x_1\leq x_2\leq x_3\leq x_4$ such that $x_2+x_3=x_1+x_4$, we have $x_2^p+x_3^p\le x_1^p+x_4^p $.
\end{proposition}
\begin{proof}
We have
\begin{align*}
    x_4^p-x_3^p
    &\geq p(x_4-x_3)x_3^{p-1}
    = p(x_2-x_1)x_3^{p-1}
    \geq p(x_2-x_1)x_2^{p-1}
    \geq x_2^p-x_1^p.
\end{align*}
The first and last inequalities above are by Proposition \ref{prop:MVT}. 
\end{proof}

\subsection{NP-hardness for constant number of terminals} 
\label{sec:np-hard-constant-many-terminals}
%NP-hardness of \mmmwc for $k=4$ terminals was shown in \cite{ST04} by a reduction from \textsc{bisection}. 
The following is the main result of this section. 

\begin{theorem}\label{thm:NP const term}
\lpnormmwc is NP-hard for every $p\ge 1$ and $k\ge 4$. 
\end{theorem}

\begin{proof}
%NP-hardness of min-max multiway cut with a constant number of terminals was shown in \cite{ST04} by a reduction from \textsc{bisection}.
%To show NP-hardness of min \lpnormmc, we use a similar reduction from \textsc{bisection}. 
We note that when $p=1$, \lpnormmc corresponds to \mwc and is known to be NP-hard for every $k\geq 3$ \cite{DJPSY94}. For the rest of our proof, we will fix $p>1$.

Our hardness reduction is from \textsc{bisection} which is known to be NP-complete. % and is an adaptation of their reduction. 
\textsc{bisection} is defined as follows:
Given an undirected graph $G=(V,E)$ where $|V|=:n$ is even and an integer $C$, the goal is to decide if there exists a subset $S\subseteq V$ such that $|S|=n/2$ and $|\delta_G(S)|\leq C$. 
%\textsc{bisection} is NP-complete. 

Given an instance $(G=(V,E),C)$ of \textsc{bisection}, we construct an instance $(G', w', T)$ of \lpnormmc consisting of a graph $G'=(V',E')$, an edge weight function $w':E'\to\R_+$, and a set $T\subseteq V'$ of $4$ terminals vertices as follows:
\begin{align*}
    V'&:=V\cup\{u,d,\ell,r\},
    \\E'&:=E\cup\{ud\}\cup\{vu,vd,v\ell,vr:v\in V\},
    \\T&:=\{u,d,\ell,r\},
    \\w'(e)&:=\begin{cases}
    1& \text{ if }e\in E
    \\a& \text{ if } e\in \{vu,vd,v\ell,vr:v\in V\}
    \\b& \text{ if } e=ud
    \end{cases},
\end{align*}
where the parameters $a$ and $b$ are given by
\begin{align*}
    a:=\max\left\{1,\frac{8n^3}{p-1},2C+1\right\},\;b:=1+\max\left\{1,(2an+C)^{\frac{p}{p-1}},3an\right\}.
\end{align*}

We note that for every fixed $p>1$, the size of $(G',w',T)$ is polynomial in the size of $(G,C)$. 
The following lemma completes the proof of the theorem. 
\end{proof}

\begin{lemma}
There exists a subset $S\subseteq V$ such that $|S|=n/2$ and $|\delta_G(S)|\le C$ if and only if $(G', w', T)$ has a multiway cut whose $\ell_p$-norm objective value is at most 
\[\left(2(b+an)^p+2(2an+C)^p\right)^{\frac{1}{p}}.\]
\end{lemma}
\begin{proof}
We start by showing the forward direction.

\begin{claim}\label{claim:constant term-1}
If there exists a subset $S\subseteq V$ such that $|S|=n/2$ and $|\delta_G(S)|\le C$, then  $(G', w', T)$ has a multiway cut whose $\ell_p$-norm objective value is at most $\left(2(b+an)^p+2(2an+C)^p\right)^{1/p}$.
\end{claim}
\begin{proof}
Let $S\subseteq V$ satisfy $|S|=n/2$ and $|\delta_G(S)|\leq C$. Then the $\ell_p$-norm objective value of the multiway cut $(\{u\},\{d\},S\cup\{\ell\},(V\backslash S)\cup\{r\})$ raised to the $p$th power is
\begin{align*}
    (b+an)^p+(b+an)^p+(3a|S|+a|V\backslash S|+C)^p&+(3a|V\backslash S|+a|S|+C)^p
    \\&=2(b+an)^p+2(2an+C)^p.
\end{align*}
\end{proof}

In order to show the reverse direction, we need the following structural result on multiway cuts of $(G',w', T)$ with cheap $\ell_p$-norm objective value. 
%at most $\left(2(b+an)^p+2(2an+C)^p\right)^{1/p}$.

\begin{claim}\label{claim:constant term-2}
If $G'$ has a multiway cut ${\mcp}$ whose $\ell_p$-norm objective value is at most $\left(2(b+an)^p+2(2an+C)^p\right)^{1/p}$, then the parts of ${\mcp}$ containing $u$ and $d$ are singletons.
\end{claim}
\begin{proof}
Let $\mcp=(U\cup\{u\},D\cup\{d\},L\cup\{\ell\},R\cup\{r\})$ be a multiway cut of $G'$ whose $\ell_p$-norm objective value raised to $p$th power is at most $2(b+an)^p+2(2an+C)^p$, where $U\cup D\cup L\cup R=V$. 
Without loss of generality, suppose that $U$ is non-empty. 
Then we have
\begin{align*}
    &w'(\delta(U\cup\{u\}))\geq b+3a|U|+a(n-|U|)=b+2a|U|+an\ge b+2a+an,
    \\&w'(\delta(D\cup\{d\}))\geq b+3a|D|+a(n-|D|)=b+2a|D|+an\ge b+an.
\end{align*}
This implies that the $\ell_p$-norm objective value of $\mcp$ raised to the $p$th power is at least $(b+2a+an)^p+(b+an)^p$. 
%If we assume at least one of $U$ and $D$ is nonempty, then we have
By assumption, the $\ell_p$-norm objective value of $\mcp$ raised to the $p$th power is at most $2(b+an)^p+2(2an+C)^p$. Thus, 
\begin{align*}
    0%&\geq (b+2a|U|+an)^p+(b+2a|D|+an)^p-(2(b+an)^p+2(2an+C)^p)\\
    &\geq (b+2a+an)^p+(b+an)^p-(2(b+an)^p+2(2an+C)^p)
    \\&=(b+2a+an)^p-(b+an)^p-2(2an+C)^p
\end{align*}
Setting $\mu(z)=z^p$, $x=b+an$ and $y=b+2a+an$ in Proposition \ref{prop:MVT}, we observe that
\begin{align*}
    (b+2a+an)^p-(b+an)^p\geq 2a\cdot \min_{z\in[b+an,b+2a+an]}pz^{p-1}=2ap(b+an)^{p-1}.
\end{align*}
Therefore,
\begin{align*}
    0&\geq (b+2a+an)^p-(b+an)^p-2(2an+C)^p
    \\&\geq 2ap(b+an)^{p-1}-2(2an+C)^p 
    \\&>2ap\left((2an+C)^{\frac{p}{p-1}}+an\right)^{p-1}-2(2an+C)^p
    \\&\geq 0.
\end{align*}
Here the strict inequality follows from our choice of $b>(2an+C)^{p/(p-1)}$ and the last inequality follows from $a\ge 1$ and $p>1$. 
This is a contradiction since one of the inequalities in the above sequence is strict. Hence, we must have $U=D=\emptyset$.
\end{proof}

The following claim proves the reverse direction of the lemma by showing that a multiway cut of $(G', w', T)$ that is cheap in $\ell_p$-norm objective value can be used to recover a cheap bisection.
\end{proof}

%Finally, to complete the backward direction, we will use the following claim to obtain a YES certificate of \textsc{bisection} from a multiway cut of $G'$ whose $\ell_p$ objective is at most $\left(2(w+an)^p+2(2an+C)^p\right)^{1/p}$.
\begin{claim}\label{claim:constant term-3}
If a multiway cut $\mcp=(U\cup\{u\},D\cup\{d\},L\cup\{\ell\},R\cup\{r\})$ of $(G', w', T)$ has $\ell_p$-norm objective value at most $\left(2(b+an)^p+2(2an+C)^p\right)^{1/p}$, then $L\cup R=V$, $|L|=|R|=n/2$, and $|\delta_G(L)|\leq C$.
\end{claim}
\begin{proof}
By Claim \ref{claim:constant term-2}, we know that $U=D=\emptyset$, and hence $L\cup R=V$. We note that in this case, we have
\begin{align*}
    &w'(\delta(U\cup\{u\}))=w'(\delta(D\cup\{d\}))=b+an,
    \\&w'(\delta(L\cup\{\ell\}))\geq 3a|L|+a(n-|L|)=2a|L|+an,
    \\&w'(\delta(R\cup\{r\}))\geq 3a|R|+a(n-|R|)=2a|R|+an.
\end{align*}
This implies that the $\ell_p$-norm objective value of $\mcp$ raised to the $p$th power is at least $2(b+an)^p+(2a|L|+an)^p + (2a|R|+an)^p$. By assumption, the $\ell_p$-norm objective value of $\mcp$ raised to the $p$th power at most $2(b+an)^p+2(2an+C)^p$. Hence, 
\begin{align}
    0&\geq 2(b+an)^p+(2a|L|+an)^p+(2a|R|+an)^p-(2(b+an)^p+2(2an+C)^p) \notag
    \\&=(2a|L|+an)^p+(2a|R|+an)^p-2(2an+C)^p. \label{eqn:hc-1}
\end{align}
For the sake of contradiction, suppose that $|L|\neq n/2$. Without loss of generality, let $|L|\geq n/2+1$ and $|R|\leq n/2-1$. We note that $|L|+|R|=n=(n/2-1)+(n/2+1)$. 

In Proposition \ref{prop:two mean values}, by setting
\begin{align*}
    x_1=2a|R|+an,\;x_2=2a\left(\frac{n}{2}-1\right)+an,\;x_3=2a\left(\frac{n}{2}+1\right)+an,\;x_4=2a|L|+an,
\end{align*}
we get
\begin{align*}
    &(2a|L|+an)^p+(2a|R|+an)^p\geq\left(2a\left(\frac{n}{2}+1\right)+an\right)^p+\left(2a\left(\frac{n}{2}-1\right)+an\right)^p.
\end{align*}
% Setting $\mu(z)=z^p$, $x=2a(n/2+1)+an$ and $y=2a|L|+an$ in Proposition \ref{prop:MVT}, we have
% \begin{align*}
%     (2a|L|+an)^p-\left(2a\left(\frac{n}{2}+1\right)+an\right)^p
%     &\geq 2a\left(|L|-\left(\frac{n}{2}+1\right)\right)\cdot \min_{z\in[2a(\frac{n}{2}+1)+an,2a|L|+an]}pz^{p-1}\\
%     &=2a\left(|L|-\left(\frac{n}{2}+1\right)\right)p\left(2a\left(\frac{n}{2}+1\right)+an\right)^{p-1}.
% \end{align*}
% Setting $\mu(z)=z^p$, $x=2a|R|+an$ and $y=2a(n/2-1)+an$ in Proposition \ref{prop:MVT}, we have
% \begin{align*}
%     \left(2a\left(\frac{n}{2}-1\right)+an\right)^p-(2a|R|+an)^p
%     &\leq 2a\left(\left(\frac{n}{2}-1\right)-|R|\right)\cdot \max_{z\in[2a|R|+an,2a(\frac{n}{2}-1)+an]}pz^{p-1}\\
%     &=2a\left(\left(\frac{n}{2}-1\right)-|R|\right)p\left(2a\left(\frac{n}{2}-1\right)+an\right)^{p-1}.
% \end{align*}
% %\knote{Something seems to be off in the $(n/2+1)$--shouldn't it be $(n/2-1)$ in the RHS?}
% These imply that
% \begin{align*}
%     &(2a|L|+an)^p+(2a|R|+an)^p-\left(2a\left(\frac{n}{2}+1\right)+an\right)^p-\left(2a\left(\frac{n}{2}-1\right)+an\right)^p 
%     \\&\geq 2a\left(|L|-\left(\frac{n}{2}+1\right)\right)p\left(2a\left(\frac{n}{2}+1\right)+an\right)^{p-1}-2a\left(\left(\frac{n}{2}-1\right)-|R|\right)p\left(2a\left(\frac{n}{2}-1\right)+an\right)^{p-1}
%     \\&\geq 0.
% \end{align*}
Substituting this in inequality \eqref{eqn:hc-1}, we get that 
\begin{align}
    0&\geq(2a|L|+an)^p+(2a|R|+an)^p-2(2an+C)^p \notag
    \\&\geq \left(2a\left(\frac{n}{2}+1\right)+an\right)^p+\left(2a\left(\frac{n}{2}-1\right)+an\right)^p-2(2an+C)^p \notag
    \\&=(2an+2a)^p+(2an-2a)^p-2(2an+C)^p \notag
    \\&=((2an+2a)^p-(2an+a)^p)+((2an+a)^p-(2an+C)^p) \notag
    \\&\qquad-((2an+C)^p-(2an)^p)-((2an)^p-(2an-2a)^p) \label{eqn:hc-2}
\end{align}
By applying Proposition \ref{prop:MVT} four times, we get that 
\begin{align*}
    &(2an+2a)^p-(2an+a)^p\geq ap(2an+a)^{p-1},
    \\&(2an+a)^p-(2an+C)^p\geq (a-C)p(2an+C)^{p-1},
    \\&(2an+C)^p-(2an)^p\leq Cp(2an+C)^{p-1},
    \\&(2an)^p-(2an-2a)^p\leq 2ap(2an)^{p-1}.
\end{align*}
Substituting these in inequality \eqref{eqn:hc-2}, we get that
\begin{align*}
    0&\geq ap(2an+a)^{p-1}+(a-C)p(2an+C)^{p-1}-Cp(2an+C)^{p-1}-2ap(2an)^{p-1}
    \\&=ap(2an+a)^{p-1}+(a-2C)p(2an+C)^{p-1}-2ap(2an)^{p-1}
    \\&\geq ap(2an+a)^{p-1}+(a-2C)p(2an)^{p-1}-2ap(2an)^{p-1}.
\end{align*}
Let $\epsilon:=(p-1)/(8n)$. Since $a\geq 8n^3/(p-1)$, we have that $2C<n^2\leq \epsilon a$. This implies
\begin{align*}
    0&\geq ap(2an+a)^{p-1}+(a-2C)p(2an)^{p-1}-2ap(2an)^{p-1}
    \\&> ap(2an+a)^{p-1}+(1-\epsilon)ap(2an)^{p-1}-2ap(2an)^{p-1}.
\end{align*}
This inequality is equivalent to
\begin{align*}
    0&>(2n+1)^{p-1}+(1-\epsilon)(2n)^{p-1}-2(2n)^{p-1}=(2n+1)^{p-1}-(1+\epsilon)(2n)^{p-1},
\end{align*}
which further implies
\begin{align*}
    \epsilon&>\left(1+\frac{1}{2n}\right)^{p-1}-1.
\end{align*}
Applying Proposition \ref{prop:MVT} again, we get
\begin{align*}
    \epsilon&>\left(1+\frac{1}{2n}\right)^{p-1}-1
    \geq\frac{1}{2n}\cdot \min_{z\in[1,1+\frac{1}{2n}]}(p-1)z^{p-2}=\frac{p-1}{2n}\min_{z\in[1,1+\frac{1}{2n}]}z^{p-2}.
\end{align*}
If $p\geq 2$, the minimum of $z^{p-2}$ for $z\in[1,1+1/(2n)]$ is attained at $z=1$, and thus $z^{p-2}\geq 1$ for all  $z\in[1,1+1/(2n)]$. If $p\in(1,2)$, the minimum of $z^{p-2}$ for $z\in[1,1+1/(2n)]$ is attained at $z=1+1/(2n)$, and thus $z^{p-2}\geq (1+1/(2n))^{p-2}>2^{p-2}>1/2$ for all  $z\in[1,1+1/(2n)]$. Hence, 
\[ \epsilon>\frac{p-1}{2n}\cdot \min_{z\in[1,1+\frac{1}{2n}]}z^{p-2}> \frac{p-1}{2n}\cdot\frac{1}{2}>\frac{p-1}{8n}=\epsilon.\]
%\knote{Why $(p-1)/2$ in the last but one inequality? Why the factor $2$? The bound here needs to be addressed in two cases: if $p\ge 2$, then the min is achieved at $z=1$. If $p\in (1, 2)$, then the min is achieved at $z=1+1/2n$ ... Can't we finish the conclusion without taking another derivative by using Taylor approximation of $(1+1/2n)^{p-1}$ or such?}
This leads to a contradiction since one of the inequalities in the above sequence is strict. Hence, we must have $|L|=|R|=n/2$. Finally, we prove the last conclusion that $|\delta_G(L)|\le C$: since $|L|=|R|=n/2$, the $\ell_p$-norm objective value of $\mcp$ raised to $p$th power is 
\[2(b+an)^p+2(2an+|\delta_G(L)|)^p\]
which is known to be at most $2(b+an)^p+2(2an+C)^p$. Hence, $|\delta_G(L)|\leq C$ as claimed.
\end{proof}

%Claims \ref{claim:constant term-1}, \ref{claim:constant term-2} and \ref{claim:constant term-3} together imply that $(G,C)$ is a YES instance of \textsc{bisection} if and only if $(G',w',T)$ has $\ell_p$ norm at most $\left(2(w+an)^p+2(2an+C)^p\right)^{1/p}$. We note that for every fixed $p>1$, the size of $(G',w',T)$ is polynomial in the size of $(G,C)$. This completes the proof of Theorem \ref{thm:NP const term}.

\subsection{NP-hardness in planar graphs}\label{sec:np-hard-planar-graphs}
%NP-hardness of \mmmwc on trees was shown in \cite{ST04} by a reduction from \textsc{3-partition}.
The following is the main result of this section. 

\begin{theorem}\label{thm:np-hard-planar}
\lpnormmwc in planar graphs is NP-hard for every $p\ge 1$. 
\end{theorem}
\begin{proof}
We note that when $p=1$, \lpnormmwc corrresponds to \mwc and is known to be NP-hard in planar graphs \cite{DJPSY94}. For the rest of our proof, we will fix $p>1$. 

%To show NP-hardness of \lpnormmc on planar graphs, we 
%also
%reduce 
Our hardness reduction is 
from \textsc{3-partition} which is known to be NP-hard. \textsc{3-partition} is defined as follows: 
Given a set $S=[3m]$, a sequence of weights $a_1,a_2,\ldots,a_{3m}$, and a bound $B$ satisfying $\sum_{i=1}^{3m}a_i=mB$ and $B/4<a_i<B/2$ for all $i\in[3m]$, the goal is to decide whether there exists a partition of $S$ into $m$ subsets $S_1,S_2,\ldots,S_m$ such that $\sum_{i\in S_j}a_i=B$ for every $j\in[m]$. %This problem is known to be NP-hard.
%\knote{Use $a_1, \ldots, a_{3m}$ to avoid confusion with our weights in the reduction.} 

Given an instance of \textsc{3-partition} by a set $S=[3m]$, weights $a_1,a_2,\ldots,a_{3m}$, and bound $B$, we construct an instance $(G,w,T)$ of \lpnormmc as follows: we start with an empty graph $G$, and for each $i\in[3m]$, we add to $G$ a subgraph as shown in Figure \ref{fig:hardness-1}. The edge weights are labelled near the corresponding edges, where $d:=(12m+12)^{\frac{1}{p-1}}$. These $3m$ subgraphs are disjoint from each other. Finally, we add $m$ isolated vertices $t_1,\ldots,t_m$ to $G$. The terminal set $T$ is given by $\{x_i^r:i\in[3m],r\in[3]\}\cup\{t_1,\ldots,t_m\}$. We observe that the graph $G$ constructed this way is planar.
\begin{figure}
    \centering
    \includegraphics[width=0.6\textwidth]{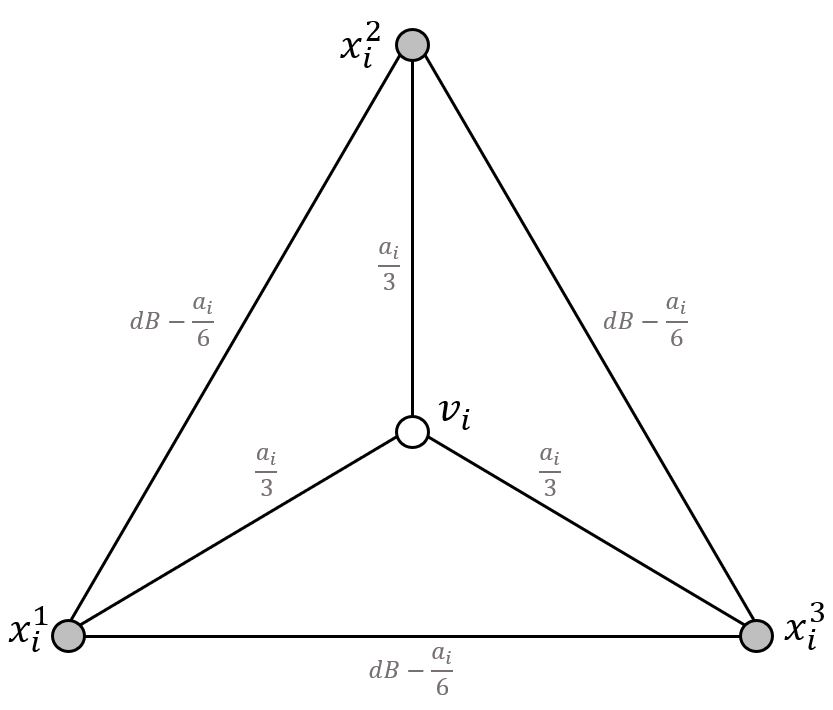}
    \caption{The $i$th subgraph in reduction from \textsc{3-partition}.}
    \label{fig:hardness-1}
\end{figure}
We note that the size of $(G,w,T)$ is polynomial in the size of the \textsc{3-partition} instance. We emphasize that the number of terminals in this reduction is not a constant. 
The following lemma completes the proof of Theorem \ref{thm:np-hard-planar}.
\end{proof}

\begin{lemma}\label{lemma:NP planar}
There exists a partition of $S$ into $S_1,\ldots,S_m$ such that $\sum_{i\in S_j}a_i=B$ for each $j\in[m]$ if and only if $(G, w, T)$ has a multiway cut whose $\ell_p$-norm objective value is at most 
\[
(9m(2dB)^p+mB^p)^{\frac{1}{p}}.
\]
\end{lemma}

\begin{proof}
We start by showing the forward direction.

\begin{claim}
If $S$ can be partitioned into $S_1,\ldots,S_m$ such that $\sum_{i\in S_j}a_i=B$ for each $j\in[m]$, then $G$ has a multiway cut whose $\ell_p$-norm objective value is $(9m(2dB)^p+mB^p)^{1/p}$.
\end{claim}

\begin{proof}
Consider the multiway cut of $G$ defined by
\[\big\{\{x_i^r\}:i\in[3m],r\in[3]\big\}\cup\big\{\{t_j\}\cup\{v_i:i\in S_j\}:j\in[m]\big\}.\]
The $p$th power of the $\ell_p$-norm objective value of this multiway cut is
\[3m\cdot 3\cdot(2dB)^p+\sum_{j\in[m]}\left(\sum_{i\in S_j}\frac{a_i}{3}\cdot 3\right)^p=9m(2dB)^p+\sum_{j\in[m]} B^p=9m(2dB)^p+mB^p.\]
\end{proof}

For the backward direction, we will start with a structural property regarding multiway cuts of $(G, w, T)$ with cheap $\ell_p$-norm objective value: each non-terminal vertex $v_i$ will not be contained in a part that contains any of the $x_i^r$ terminals. 

\begin{claim}\label{claim:np-hardness-planar-1}
Let $\mcp$ be a multiway cut in $(G, w, T)$ with $\ell_p$-norm objective value at most $(9m(2dB)^p+mB^p)^{1/p}$. Then, for every $i\in [3m]$, the vertex $v_i$ will be in a part of $\mcp$ that contains terminal $t_j$ for some $j\in [m]$.
\end{claim}
\begin{proof}
%Consider an optimal multiway cut of $G$ whose $\ell_p$-norm objective value is at most $9m(2dB)^p+mB^p$. 
We will use $X_i^r$ to denote the part in $\mcp$ containing $x_i^r$ for each $i\in[3m]$ and $r\in[3]$, and $T_j$ to denote the part in $\mcp$ containing $t_j$ for each $j\in[m]$.
%For convenience, we will define $I\subseteq[3m]$ by
Let $I:=\left\{i\in[3m]:v_{i}\in\bigcup_{i'\in[3m],r\in[3]}X_{i'}^r\right\}$ be the indices of vertices $v_i$ that are contained in some part that contains an $x_i^r$ terminal. Suppose for the sake of contradiction that $I\neq \emptyset$. 

Let $i\in[3m],r\in[3]$. Then, we have
\begin{align}
w(\delta(X_i^r))\geq 2dB+\sum_{i'\in[3m]:v_{i'}\in X_i^r}\frac{a_{i'}}{3}.\label{eqn:hp-1}
\end{align}
For an arbitrarily fixed $X_i^r$ such that $|X_i^r|\geq 3$, let $X_i^r=\{x_i^r,v_{i_1},v_{i_2},\ldots,v_{i_\ell}\}$ for some $i_1,\ldots,i_\ell\in[3m]$, where $\ell\geq 2$. In Proposition \ref{prop:two mean values}, if we choose
\begin{align*}
x_1&=2dB,\ x_4=2dB+\sum_{q=1}^\ell \frac{a_{i_q}}{3},\\
x_2&=\min\left\{2dB+\frac{a_{i_\ell}}{3},\,2dB+\sum_{q=1}^{\ell-1} \frac{a_{i_q}}{3}\right\},
\\x_3&=\max\left\{2dB+\frac{a_{i_\ell}}{3},\,2dB+\sum_{q=1}^{\ell-1} \frac{a_{i_q}}{3}\right\}, 
\end{align*}
%\knote{need to ensure that $x_2 \le x_3$?}\wnote{My previous notation meant that $x_2$ is the smaller of the two, and $x_3$ is the greater of the two}
then we have
\begin{align*}
    \left(2dB+\sum_{q=1}^\ell \frac{a_{i_q}}{3}\right)^p+(2dB)^p\geq \left(2dB+\sum_{q=1}^{\ell-1} \frac{a_{i_q}}{3}\right)^p+\left(2dB+\frac{a_{i_\ell}}{3}\right)^p.
\end{align*}
By applying this argument $\ell-1$ times, we get
\begin{align}
    &\left(2dB+\sum_{q=1}^\ell \frac{a_{i_q}}{3}\right)^p+(\ell-1)(2dB)^p\notag
    \\&\geq \left(2dB+\sum_{q=1}^{\ell-1} \frac{a_{i_q}}{3}\right)^p+\left(2dB+\frac{a_{i_\ell}}{3}\right)^p+(\ell-2)(2dB)^p\notag
    \\&\geq\ldots\notag
    \\&\geq \left(2dB+\sum_{q=1}^{2} \frac{a_{i_q}}{3}\right)^p+\left(2dB+\frac{a_{i_3}}{3}\right)^p+\ldots+\left(2dB+\frac{a_{i_\ell}}{3}\right)^p+(2dB)^p\notag
    \\&\geq \sum_{q=1}^\ell\left(2dB+\frac{a_{i_q}}{3}\right)^p.\label{eqn:hp-2}
\end{align}

We will divide parts $X_i^r$ into three categories by defining
\begin{align*}
    \mcx_1&:=\left\{X_i^r:i\in[3m],r\in[3],|X_i^r|=1\right\},
    \\\mcx_2&:=\left\{X_i^r:i\in[3m],r\in[3],|X_i^r|=2\right\},
    \\\mcx_3&:=\left\{X_i^r:i\in[3m],r\in[3],|X_i^r|\geq 3\right\}.
\end{align*}
Moreover, let us define two subsets of $I$ by
\begin{align*}
    I_2:=\left\{i'\in I:v_{i'}\in\bigcup_{X_i^r\in\mcx_2}X_i^r\right\},\;I_3:=\left\{i'\in I:v_{i'}\in\bigcup_{X_i^r\in\mcx_3}X_i^r\right\}.
\end{align*}
We note that $(I_2,I_3)$ form a partition of $I$, $|I_2|=|\mcx_2|$, and $|\mcx_1|+|\mcx_2|+|\mcx_3|=9m$.

Then the contribution of sets $X_i^r$ to the $p$th power of the $\ell_p$-norm objective value is given by
\begin{align}
    \sum_{i\in[3m],r\in[3]}w(\delta(X_i^r))^p&=\sum_{X_i^r\in \mcx_1}w(\delta(X_i^r))^p+\sum_{X_i^r\in \mcx_2}w(\delta(X_i^r))^p+\sum_{X_i^r\in \mcx_3}w(\delta(X_i^r))^p\notag
    \\&= |\mcx_1|(2dB)^p+\sum_{X_i^r\in \mcx_2}w(\delta(X_i^r))^p+\sum_{X_i^r\in \mcx_3}w(\delta(X_i^r))^p.\label{eqn:hp-3}
\end{align}

By applying \eqref{eqn:hp-1} and \eqref{eqn:hp-2} to members of $\mcx_3$, we get
\begin{align*}
    &\sum_{X_i^r\in \mcx_3}w(\delta(X_i^r))^p\geq\sum_{X_i^r\in \mcx_3} \left(2dB+\sum_{i'\in[3m]:v_{i'}\in X_i^r}\frac{a_{i'}}{3}\right)^p \quad \quad \quad \quad \text{(by \eqref{eqn:hp-1})}
    \\&=\sum_{X_i^r\in \mcx_3} \left(\left(2dB+\sum_{i'\in[3m]:v_{i'}\in X_i^r}\frac{a_{i'}}{3}\right)^p+(|X_i^r|-2)(2dB)^p\right)-\sum_{X_i^r\in\mcx_3}(|X_i^r|-2)(2dB)^p
    \\&\geq \sum_{X_i^r\in \mcx_3}\sum_{i'\in[3m]:v_{i'}\in X_i^r}\left(2dB+\frac{a_{i'}}{3}\right)^p-\sum_{X_i^r\in\mcx_3}(|X_i^r|-2)(2dB)^p \quad \quad \quad \quad  \text{(by \eqref{eqn:hp-2})}
    \\&=\sum_{i'\in I_3}\left(2dB+\frac{a_{i'}}{3}\right)^p-\sum_{X_i^r\in\mcx_3}(|X_i^r|-1)(2dB)^p+|\mcx_3|(2dB)^p
    \\&=\sum_{i'\in I_3}\left(2dB+\frac{a_{i'}}{3}\right)^p-|I_3|(2dB)^p+|\mcx_3|(2dB)^p.
\end{align*}

Moreover, we observe that
\begin{align*}
    \sum_{X_i^r\in \mcx_2}w(\delta(X_i^r))^p&\geq\sum_{X_i^r\in \mcx_2} \left(2dB+\sum_{i'\in[3m]:v_{i'}\in X_i^r}\frac{a_{i'}}{3}\right)^p=\sum_{i'\in I_2}\left(2dB+\frac{a_{i'}}{3}\right)^p.
\end{align*}

Therefore, \eqref{eqn:hp-3} implies that
\begin{align*}
    \sum_{i\in[3m],r\in[3]}&w(\delta(X_i^r))^p=|\mcx_1|(2dB)^p+\sum_{X_i^r\in \mcx_2}w(\delta(X_i^r))^p+\sum_{X_i^r\in \mcx_3}w(\delta(X_i^r))^p
    \\&\geq|\mcx_1|(2dB)^p+\sum_{i'\in I_2}\left(2dB+\frac{a_{i'}}{3}\right)^p+\sum_{i'\in I_3}\left(2dB+\frac{a_{i'}}{3}\right)^p-|I_3|(2dB)^p+|\mcx_3|(2dB)^p
    \\&=|\mcx_1|(2dB)^p+\sum_{i'\in I}\left(2dB+\frac{a_{i'}}{3}\right)^p-|I_3|(2dB)^p+|\mcx_3|(2dB)^p
    \\&=\sum_{i'\in I}\left(2dB+\frac{a_{i'}}{3}\right)^p+(|\mcx_1|+|\mcx_3|)(2dB)^p-|I_3|(2dB)^p
    \\&=\sum_{i'\in I}\left(2dB+\frac{a_{i'}}{3}\right)^p+(9m-|\mcx_2|)(2dB)^p-|I_3|(2dB)^p
    \\&=\sum_{i'\in I}\left(2dB+\frac{a_{i'}}{3}\right)^p+(9m-|I_2|)(2dB)^p-|I_3|(2dB)^p
    \\&=\sum_{i'\in I}\left(2dB+\frac{a_{i'}}{3}\right)^p+(9m-|I|)(2dB)^p
    \\&\geq \sum_{i'\in I}\left(2dB+\frac{B}{12}\right)^p+(9m-|I|)(2dB)^p \quad \quad \text{(since $a_{i'}\ge \frac{B}{4}$ $\forall i'\in [3m]$)}
    \\&=|I|\left(2d+\frac{1}{12}\right)^p B^p+(9m-|I|)(2dB)^p.
\end{align*}

% By repeatedly applying this inequality to all $X_i^r$ such that $|X_i^r|\geq 2$, we have \knote{Are you sure that it is $9m-|I|$ below? Explain this. }
% \knote{You might want to rewrite the inequality from the previous series as $(...)^p \ge \sum(...)^p - (\ell-1)(...)$. I am not sure what happens to the $-1$ term of $(\ell-1)$ in your next series of inequalities. }
% \begin{align*}
%     \sum_{i\in[3m],r\in[3]}&w(\delta(X_i^r))^p=\sum_{\substack{i\in[3m],r\in[3]:\\|X_i^r|\geq 2}}\left(2dB+\sum_{i'\in[3m]:v_{i'}\in X_i^r}\frac{a_{i'}}{3}\right)^p+\sum_{\substack{i\in[3m],r\in[3]:\\|X_i^r|=1}}(2dB)^p
%     \\&=\sum_{\substack{i\in[3m],r\in[3]:\\|X_i^r|\geq 2}}\left(\left(2dB+\sum_{i'\in[3m]:v_{i'}\in X_i^r}\frac{a_{i'}}{3}\right)^p-(|X_i^r|-1)(2dB)^p\right)+(9m-|I|)(2dB)^p
%     \\&\geq\sum_{\substack{i\in[3m],r\in[3]:\\|X_i^r|\geq 2}}\sum_{\substack{i'\in[3m]:\\v_{i'}\in X_i^r}}\left(2dB+\frac{a_{i'}}{3}\right)^p+(9m-|I|)(2dB)^p
%     % \\&\geq \sum_{i\in[3m],r\in[3]}\left(2dB+\sum_{i'\in[3m]:v_{i'}\in X_i^r}\frac{a_{i'}}{3}\right)^p
%     \\&=\sum_{i\in I}\left(2dB+\frac{a_i}{3}\right)^p+(9m-|I|)(2dB)^p
%     \\&\geq \sum_{i\in I}\left(2dB+\frac{B}{12}\right)^p+(9m-|I|)(2dB)^p
%     \\&=|I|\left(2d+\frac{1}{12}\right)^p B^p+(9m-|I|)(2dB)^p.
% \end{align*}
Assuming $I\neq\emptyset$, the $p$th power of the $\ell_p$-norm objective value of this multiway cut is 
\begin{align*}
    &\sum_{i\in[3m],r\in[3]}w(\delta(X_i^r))^p+\sum_{j=1}^mw(\delta(T_j))^p
    \\&\geq |I|\left(2d+\frac{1}{12}\right)^p B^p+(9m-|I|)(2dB)^p+\sum_{i\notin I}w(\delta(v_i))^p
    \\&=|I|\left(2d+\frac{1}{12}\right)^p B^p+(9m-|I|)(2dB)^p+\sum_{i\notin I}a_i^p
    \\&\geq |I|\left(2d+\frac{1}{12}\right)^p B^p+(9m-|I|)(2dB)^p+(3m-|I|)\left(\frac{B}{4}\right)^p
    \\&=|I|\left(\left(2d+\frac{1}{12}\right)^p-(2d)^p-\left(\frac{1}{4}\right)^p\right)B^p+9m(2d)^pB^p+3m\left(\frac{1}{4}\right)^pB^p.
\end{align*}
Since we assumed that the $\ell_p$-norm objective value of this multiway cut is at most $(9m(2dB)^p+mB^p)^{1/p}$, we have
\begin{align*}
    0&\geq \left(\sum_{i\in[3m],r\in[3]}w(\delta(X_i^r))^p+\sum_{j=1}^mw(\delta(T_j))^p-(9m(2dB)^p+mB^p)\right)B^{-p}
    \\&\geq |I|\left(\left(2d+\frac{1}{12}\right)^p-(2d)^p-\left(\frac{1}{4}\right)^p\right)+3m\left(\frac{1}{4}\right)^p-m
    \\&>|I|\left(\left(2d+\frac{1}{12}\right)^p-(2d)^p-\left(\frac{1}{4}\right)^p\right)-m.
\end{align*}
We note that due to Proposition \ref{prop:MVT}, we have
\begin{align*}
    \left(2d+\frac{1}{12}\right)^p-(2d)^p\geq \frac{p(2d)^{p-1}}{12}=2^{p-1}\cdot p(m+1).
\end{align*}
This implies
\begin{align*}
    0>|I|\left(2^{p-1}\cdot p(m+1)-\left(\frac{1}{4}\right)^p\right)-m\geq 2^{p-1}\cdot p(m+1)-\left(\frac{1}{4}\right)^p-m \geq 0,
\end{align*}
yielding a contradiction since one of the inequalities in the sequence is strict. Therefore, we must have $I=\emptyset$.

%For each $i\in[3m]$ and $r\in[3]$, if $X_i^r$ is not singleton, then we have $w(\delta(X_i^r))> 2dB+B/12$. To see this, let us first consider the case of $X_i^r\supseteq\{x_i^r,v_i\}$. In this case, we have $w(\delta(X_i^r))\geq w(\delta(\{x_i^r,v_i\}))> 2dB+B/12$. Then let us consider the case of $X_i^r\supseteq\{x_i^r,v_{i'}\}$ for some $i'\neq i$ and $v_i\notin X_i^r$. In this case, we have $w(\delta(X_i^r))\geq 2dB+a_j> 2dB+B/12$ as claimed as well. Moreover, if $X_i^r$ is singleton, then $w(\delta(X_i^r))=2dB$.

% For convenience, we will define $I\subseteq[3m]$ by $I:=\{i'\in[3m]:v_{i'}\in\bigcup_{i\in[3m],r\in[3]}X_i^r\}$. Then by convexity, we have
% \[\sum_{j=1}^m w(\delta(S_j))^p\geq \sum_{i\notin I}w(\delta(v_i))^p=\sum_{i\notin I}a_i^p>(3m-|I|)(B/4)^p.\]

% Therefore, assuming $I\neq\emptyset$, the $\ell_p$ objective of this multiway cut is at least
% \begin{align*}
%     &\sum_{i\in[3m],r\in[3]}w(\delta(X_i^r))^p+\sum_{j=1}^mw(\delta(S_j))^p
%     \\&>(2dB+B/12)^p+(9m-1)(2dB)^p+(3m-|I|)(B/4)^p
% \end{align*}
\end{proof}

We complete the proof of the backward direction by showing the following claim which derives a YES certificate for \textsc{3-partition} from an optimal multiway cut whose $\ell_p$-norm objective value is at most $(9m(2dB)^p+mB^p)^{1/p}$.
\end{proof}

\begin{claim}
Given a multiway cut $\mcp$ of $(G,w, T)$ whose $\ell_p$-norm objective value is at most $(9m(2dB)^p+mB^p)^{1/p}$, let $S_j:=\{i\in[3m]:v_i\text{ is in the same part as }t_j \text{ in } \mcp\}$ for each $j\in[m]$. Then $(S_1,S_2,\ldots,S_m)$ is a partition of $S=[3m]$ such that $\sum_{i\in S_j}a_i=B$ for each $j\in[m]$.
\end{claim}

\begin{proof}
By Claim \ref{claim:np-hardness-planar-1}b, we know that $S_1,\ldots,S_j$ must form a partition of $S$. This also implies that for each $i\in[3m]$, $r\in[3]$, the set $\{x_i^r\}$ is a part in the multiway cut $\mcp$. Therefore, the $\ell_p$-norm objective value of $\mcp$ is at least 
we have
\begin{align*}
    %9m(2dB)^p+mB^p
    &\sum_{i\in[3m],r\in [3]}(2dB)^p+\sum_{j=1}^mw(\delta(S_j))^p
    =9m(2dB)^p+\sum_{j=1}^mw(\delta(S_j))^p.
\end{align*}
Since we know that the $\ell_p$-norm objective value of $\mcp$ is at most $9m(2dB)^p+mB^p$, it follows that $\sum_{j=1}^mw(\delta(S_j))^p\leq mB^p$. By Jensen's inequality, we observe that
\begin{align*}
    mB^p\geq\sum_{j=1}^mw(\delta(S_j))^p=\sum_{j=1}^m \left(\sum_{i\in S_j}a_i\right)^p\geq m\cdot\left(\frac{1}{m}\sum_{j=1}^m\sum_{i\in S_j}a_i\right)^p=mB^p.
\end{align*}
Hence, all inequalities above should be equations. This happens only when $\sum_{i\in S_j}a_i=B$ for all $j\in [m]$.
\end{proof}

%\knote{NP-hardness for constant number of terminals?}

\section{Convex program and integrality gap}\label{sec:convex-program}

%\begin{theorem}\label{lem:integrality-gap}
%The convex programming relaxation of min \lpnormmc has an integrality gap at least $k^{1-1/p}/2$.
%\end{theorem}

%Min \lpnormmc on input $(G=(V,E),w,T=\{t_1,\ldots,t_k\})$ has the following convex programming relaxation:

The following is a natural convex programming relaxation for \lpnormmwc on instance $(G, w, T)$ where $T=\{t_1, \ldots, t_k\}$ are the terminal vertices (the objective function can be convexified by introducing additional variables and constraints): 
\begin{align}
    \text{Minimize }&\left(\sum_{i=1}^k\left(\sum_{uv\in E}w(uv)\cdot |x(u,i)-x(v,i)|\right)^p\right)^{1/p}\text{ subject to} \label{cp}\\
    \sum_{i=1}^k x(v,i)&=1\quad\forall v\in V,\notag\\
    x(t_i,i)&=1\quad\forall i\in[k],\notag\\
    x(v,i)&\geq 0\quad\forall v\in V,\ \forall i\in[k].\notag
\end{align}

\begin{lemma}\label{lem:integrality-gap}
The convex program in \eqref{cp} has an integrality gap of at least $k^{1-1/p}/2$. 
\end{lemma}
\begin{proof}
Consider the star graph that has $k$ leaves $\{t_1, \ldots, t_k\}$ and a center vertex $v$ with all edge weights being $1$. Let the terminal vertices be the $k$ leaves. 
%instance in Figure \ref{fig:convex programming-1}. 
%The minimum $L_p$-norm of a multiway cut is 
The optimum $\ell_p$-norm objective value of a multiway cut is 
\[((k-1)^p+k-1)^{\frac{1}{p}},\]
and it corresponds to the partition $(\{t_1,v\}, \{t_2\}, \{t_3\}, \ldots, \{t_k\})$. 
A feasible solution to the convex program \eqref{cp} is given by $x(v,i)=1/k$ for all $i\in[k]$, which yields an objective of
\[\left(k\cdot\left(\frac{k-1}{k}+(k-1)\cdot\frac{1}{k}\right)^p\right)^{\frac{1}{p}}=\frac{2k-2}{k}\cdot k^{\frac{1}{p}}.\]
This results in an integrality gap of at least
\begin{align*}
    \frac{((k-1)^p+k-1)^{\frac{1}{p}}}{\frac{2k-2}{k}\cdot k^{\frac{1}{p}}}\geq \frac{k-1}{\frac{2k-2}{k}\cdot k^{\frac{1}{p}}}=\frac{k^{1-\frac{1}{p}}}{2}.
\end{align*}
\end{proof}

Bansal et al. give an SDP relaxation for \mmmwc and show that the star graph has an integrality gap of $\Omega(k)$ for this SDP relaxation. This SDP relaxation can be generalized in a natural fashion to \lpnormmwc. The star graph still exhibits an integrality gap of $\Omega(k^{1-1/p})$ for the generalized SDP relaxation for \lpnormmwc. 

%The SDP relaxation in \cite{BFKMNNS14} for min-max multiway cut has an integrality gap of $k^{1-1/p}/2$ when applied to min $L_p$-norm multiway cut as well. \knote{How do you adapt the SDP for \lpnormmc? Write the SDP.}
\iffalse
\begin{figure}
    \centering
    \includegraphics[width=0.4\textwidth]{Convex-Programming-1.JPG}
    \caption{An instance with large integrality gap. All edge weights are uniformly $1$.}
    \label{fig:convex programming-1}
\end{figure}
\fi

\section{Inapproximability}
\label{sec:inapprox}

In this section, we show that \lpnormmwc does not admit a $k^{1-1/p-\epsilon}$-approximation assuming the small set expansion hypothesis. In contrast, there is a trivial $O(k^{1-1/p})$-approximation (see Section \ref{sec:trivial-approx}). 
We mention that the inapproximability result in this section is 
%our proof of Theorem \ref{thm: mskp bic} is 
similar to the result of Bansal et al \cite{BFKMNNS14} who showed that \mmmwc does not admit a $k^{1-\epsilon}$-approximation assuming the small set expansion hypothesis. We adapt the same ideas for \lpnormmwc. 

To prove our results, we consider \mskp: the input to this problem is a graph $G=(V,E)$ (where $n:=|V|$), an edge weight function $w: E\rightarrow \R_+$, and an integer $k\le n$. The goal is to partition $V$ into $k$ sets $P_1, \ldots, P_k$ such that $|P_i|=n/k$ for all $i\in [k]$ so as to minimize $\sum_{i=1}^k w(\delta(P_i))$. We will use $\lambda$ to denote the optimum objective value of \mskp. A partition $(P_1, \ldots, P_k)$ of $V$ is a $(\alpha, \beta)$-bicriteria approximation for \mskp if  $\sum_{i=1}^k w(\delta(P_i)) \le \alpha \lambda$ and $|P_i|\le \beta(n/k)$ for all $i\in [k]$. 
For constant $k$, it is known that a $(O(1), O(1))$-bicriteria approximation for \mskp is at least as hard as small set expansion \cite{RST12}. 
We show the following result which implies that a $k^{1-1/p-\epsilon}$-approximation is unlikely for \lpnormmwc (by setting $k=k(\epsilon)$ to be a large constant):
\begin{theorem}\label{thm: mskp bic}
If \lpnormmc admits an efficient $k^{1-1/p-\epsilon}$-approximation algorithm for some constant $\epsilon>0$, then \mskp admits a $(O(k^{2-1/p}),O(1))$-bicriteria approximation for sufficiently large $k$.
\end{theorem}

%Bansal et al show that approximating \mmmwc to within $k^{1-\epsilon}$ factor for every $\epsilon>0$ is SSE-hard. We will use a similar argument to show that approximating \mmmwc to within $k^{1-1/p-\epsilon}$ factor for every $\epsilon>0$ is SSE-hard.

%This argument necessitates a problem termed as \mskp. \mskp takes as input a graph $G=(V,E)$ (where $|V|=:n$), an edge weight function $w:E\to\R_+$, and an integer $k\leq n$. The goal is to partition $V$ into $k$ sets $P_1,\ldots,P_k$ such that $|P_i|=n/k$ for all $i\in[k]$ and minimize $\sum_{i=1}^k w(\delta_G(P_i))$. We will use $\lambda$ to refer to the optimum objective value of \mskp on instance $(G,w,k)$. A partition $(P_1,\ldots,P_k)$ of $V$ is a $(\alpha, \beta)$-bicriteria approximation for \mskp if $\sum_{i=1}^k w(\delta_G(P_i))\leq\alpha\lambda$ and $|P_i|\leq \beta n/k$ for all $i\in[k]$.

%\begin{theorem}\label{thm: mskp bic}
%If \lpnormmc admits an efficient $k^{1-1/p-\epsilon}$-approximation algorithm for some $\epsilon>0$, then \mskp admits a $(O(k^{2-1/p}),O(1))$-bicriteria approximation for sufficiently large $k$.
%\end{theorem}

%In this section, we prove theorem \ref{thm: mskp bic}. 
Our proof of Theorem \ref{thm: mskp bic} proceeds via the following lemma (which is the counterpart to Lemma 5.1 of \cite{BFKMNNS14}, but for \lpnormmwc).

\begin{lemma}\label{lem:inapprox}
If \lpnormmc has a polynomial-time $\gamma$-approximation algorithm, then \mskp has an efficient $(5k\gamma,9\gamma k^{1/p})$-bicriteria approximation algorithm.
\end{lemma} 
%\knote{This should be a lemma and not a theorem. What is the ultimate conclusion from this lemma? Is it about the inapproximability of \lpnormmwc? Complete the details and derive that conclusion---state that as the main theorem. } 

%\knote{This conclusion would make better sense than the current form.}

%\knote{Put the details below in proof environment. Follow similar writing style as the NP-hardness sections.} 

\begin{proof}
We will follow the reduction designed by Bansal et al in Lemma 5.1 of \cite{BFKMNNS14}. Let $(G,w,k)$ be an instance of \mskp, and let $\lambda$ refer to the optimum objective value of \mskp on instance $(G,w,k)$. We will assume knowledge of a value $B\in[\lambda,2\lambda]$ by binary search. We construct an instance $(G'=(V',E'),w',T)$ of \lpnormmc as follows. 
\begin{align*}
    &V':=V\cup\{t_1,\ldots,t_k\},
    \\&E':=E\cup\{t_iv:v\in V\},
    \\&T:=\{t_1,\ldots,t_k\},
    \\&w'(e):=\begin{cases}
    w(e)& \text{ if }e\in E, \\ 
    \frac{B}{n}& \text{ if } e\in E'\setminus E.
    \end{cases}
\end{align*}

We will use $\OPT$ to refer to the optimum $\ell_p$-norm objective value of \lpnormmc on instance $(G',w',T)$. The following claim completes the proof of Lemma \ref{lem:inapprox}.
\end{proof}

\begin{claim}
If $\mcp'=(P'_1,\ldots,P'_k)$ is a multiway cut on instance $(G',w',T)$ such that with $\ell_p$-norm objective value at most $\gamma\cdot\OPT$, then the partition $\mcp=(P_1,\ldots,P_k)$ of $V$ defined by $P_i=P'_i\cap V$ for all $i\in[k]$ is a $(5k\gamma,9\gamma k^{1/p})$-bicriteria approximate optimum to \mskp.
\end{claim}
\begin{proof}
Let $\mcq=(Q_1,\ldots,Q_k)$ be an optimum solution to \mskp on instance $(G,w,k)$, and let $\mcq'=(Q'_1,\ldots,Q'_k)$ be a partition of $V(G')$ obtained by $Q'_i:=Q_i\cup\{t_i\}$ for each $i\in[k]$. Then, $\mcq'$ is a multiway cut for $(G', w', T)$ and the $\ell_p$-norm objective value of $\mcq'$ raised to $p$th power is
\begin{align}
    \sum_{i=1}^k w'(\delta_{G'}(Q'_i))^p&=\sum_{i=1}^k\left(w(\delta_G(Q_i))+\frac{n}{k}\cdot(k-1)\cdot\frac{B}{n}+\frac{n}{k}\cdot(k-1)\cdot\frac{B}{n}\right)^p \notag    \\
    &=\sum_{i=1}^k\left(w(\delta_G(Q_i))+\frac{k-1}{k}\cdot 2B\right)^p\geq \OPT^p, \label{eq:inapprox-1}
\end{align}
where the first $(n/k)(k-1)(B/n)$ term represents the cost of edges between $Q_i$ and $T-\{t_i\}$, and the second $(n/k)(k-1)(B/n)$ term represents the cost of edges between $V-Q_i$ and $t_i$.

Since $\mcp'$ is a $\gamma$-approximate optimum solution to \lpnormmc, we have
\begin{align}
    \gamma^p\cdot\OPT^p&\geq \sum_{i=1}^k w'(\delta_{G'}(P'_i))^p\notag
    \\&=\sum_{i=1}^k\left(w(\delta_G(P_i))+|P_i|(k-1)\frac{B}{n}+(n-|P_i|)\frac{B}{n}\right)^p\notag
    \\&=\sum_{i=1}^k\left(w(\delta_G(P_i))+B+(k-2)|P_i|\frac{B}{n}\right)^p \label{eqn:inapprox-1}
    \\&\geq k^{1-p}\left(\sum_{i=1}^k w(\delta_G(P_i))+(2k-2)B\right)^p.\qquad\text{(by Jensen's inequality)}\notag
\end{align}
Hence, 
\begin{align}
    k^{1-p}\left(\sum_{i=1}^k w(\delta_G(P_i))+(2k-2)B\right)^p
    &\leq \gamma^p\cdot\OPT^p\notag
    \\&\leq \gamma^p\sum_{i=1}^k\left(w(\delta_G(Q_i))+\frac{k-1}{k}\cdot 2B\right)^p \quad \quad \text{(by \eqref{eq:inapprox-1})}\notag
    \\&\leq \gamma^p\sum_{i=1}^k\left(\lambda+\frac{k-1}{k}\cdot 2B\right)^p\notag
    \\&=\gamma^pk\left(\lambda+\frac{k-1}{k}\cdot 2B\right)^p.\label{eqn:inapprox-2}
\end{align}
This inequality is equivalent to
\[\sum_{i=1}^k w(\delta_G(P_i))+(2k-2)B\leq k\gamma\left(\lambda+\frac{k-1}{k}\cdot 2B\right).\]
Combining the assumption that $B\in[\lambda,2\lambda]$, we have
\begin{align*}
    \sum_{i=1}^k w(\delta_G(P_i))\leq k\gamma\left(\lambda+\frac{k-1}{k}\cdot 2B\right)<k\gamma(\lambda+2B)\leq 5k\gamma\lambda.
\end{align*}
Inequalities \eqref{eqn:inapprox-1} and \eqref{eqn:inapprox-2} also imply that for every $j\in[k]$, 
\begin{align*}
    \left((k-2)|P_j|\frac{B}{n}\right)^p&\leq
    \sum_{i=1}^k\left(w(\delta_G(P_i))+B+(k-2)|P_i|\frac{B}{n}\right)^p
    \\&\leq\gamma^p\cdot\OPT^p\qquad\text{(by \eqref{eqn:inapprox-1})}
    \\&\leq \gamma^p\cdot k\left(\lambda+\frac{k-1}{k}\cdot 2B\right)^p\qquad\text{(by \eqref{eqn:inapprox-2})}
    \\&\leq \gamma^p\cdot k(3B)^p.
\end{align*}
This implies that 
\[|P_j|\leq 3\gamma k^{\frac{1}{p}}\frac{n}{k-2}\leq 9\gamma k^{\frac{1}{p}}\frac{n}{k}.\]
\end{proof}

We will use the following lemma from \cite{BFKMNNS14} to prove Theorem \ref{thm: mskp bic}.

\begin{lemma}\cite{BFKMNNS14}\label{lem: 5.3 of BFKMNNS}
If \mskp has an efficient $(\alpha,k^{1-\epsilon})$-bicriteria approximation algorithm for some $\epsilon>0$, then \mskp also has an efficient $(\alpha\log\log k,3^{2/\epsilon})$-bicriteria approximation algorithm.
\end{lemma}

We now prove Theorem \ref{thm: mskp bic}. 
\begin{proof}[Proof of Theorem \ref{thm: mskp bic}]
If \lpnormmc has an efficient $k^{1-1/p-\epsilon}$-approximation algorithm for some $\epsilon>0$, then Lemma \ref{lem:inapprox} implies that \mskp has an efficient $(5k^{2-1/p-\epsilon},9k^{1-\epsilon})$-bicriteria approximation algorithm. For $k$ sufficiently large, we have $9k^{1-\epsilon}\leq k^{1-\epsilon/2}$. Lemma \ref{lem: 5.3 of BFKMNNS} then implies that \mskp has a $(O(k^{2-1/p}),3^{4/\epsilon})$-bicriteria approximation. This completes the proof of Theorem \ref{thm: mskp bic}.
\end{proof}

\subsection{A trivial $O(k^{1-1/p})$-approximation}\label{sec:trivial-approx}
In this section, we show a trivial approximation algorithm for \lpnormmwc that achieves an approximation factor of $O(k^{1-1/p})$. Given an instance $(G,w,T)$ of \lpnormmwc, let the set $T$ of terminals be $\{t_1, \ldots, t_k\}$. For each $i\in [k]$, we compute a minimum $(t_i, T-t_i)$-cut, say $(S_i, V-S_i)$. The sets $S_1, \ldots, S_k$ can be uncrossed via posimodularity to ensure that each $(S_i, V-S_i)$ is still a minimum $(t_i, T-T_i)$-cut and moreover $S_i \cap S_j = \emptyset$ for all distinct $i, j\in [k]$. Let $R:=V-\cup_{i=1}^k S_i$. We will show that the multiway cut $(S_1 \cup R, S_2, S_3, \ldots, S_k)$ is a $O(k^{1-1/p})$-approximation for \lpnormmwc. 

Let $(P_1, \ldots, P_k)$ denote an optimum solution for \lpnormmwc. Since $(S_i, V-S_i)$ is a min $(t_i, T-t_i)$-cut, we have that $w(\delta(S_i))\le w(\delta(P_i))$. We also note that $w(\delta(S_1\cup R)) \le w(\delta(S_1))+ w(\delta(R)) \le 2\sum_{i=1}^k w(\delta(S_i))$ since $\delta(R) \subseteq  \cup_{i=1}^k \delta(S_i)$. 
Let us consider the $p$th power of the $\ell_p$-norm objective value of $(S_1\cup R, S_2, \ldots, S_k)$:
\begin{align*}
    w(\delta(S_1\cup R))^p + \sum_{i=2}^k w(\delta(S_i))^p
    & \le \left(2\sum_{i=1}^k w(\delta(S_i))\right)^p + \sum_{i=2}^k w(\delta(S_i))^p \\
    & \le 2^p k^{p-1} \sum_{i=1}^k w(\delta(S_i))^p + \sum_{i=2}^k w(\delta(S_i))^p \quad \quad \quad \quad \text{(by Jensen)}\\
    & \le 2^p k^{p-1} \sum_{i=1}^k w(\delta(S_i))^p\\
    & \le 2^p k^{p-1} \sum_{i=1}^k w(\delta(P_i))^p. 
\end{align*}
Hence, the $\ell_p$-norm objective value of $(S_1\cup R, S_2, \ldots, S_k)$ is within a $(2k^{1-1/p})$-factor of the optimum $\ell_p$-norm objective value.

\section{Conclusion}
\label{sec:conclusion}
In this work, we introduced \lpnormmwc for $p\ge 1$ as a unified generalization of \mwc and \mmmwc. We showed that \lpnormmwc is NP-hard for constant number of terminals or in planar graphs for every $p\ge 1$. The natural convex program for \lpnormmwc has an integrality gap of $\Omega(k^{1-1/p})$ and the problem is $(k^{1-1/p-\epsilon})$-inapproximable for any constant $\epsilon>0$ assuming the small set expansion hypothesis, where $k$ is the number of terminals in the input instance. The inapproximability result suggests that a dependence on $n$ in the approximation factor is unavoidable if we would like to obtain an approximation factor that is better than the trivial $O(k^{1-1/p})$-factor. 
%, thus suggesting that a dependence on $n$ is unavoidable if we need a non-trivial approximation factor. 
On the algorithmic side, we gave an 
%$O(\log^{1.5} n \log^{0.5}{k})$-approximation 
$O(\sqrt{\log^{3} n \log{k}})$-approximation 
(i.e., an $O(\log^2{n})$-approximation), where $n$ is the number of vertices in the input graph. Our results suggest that the approximability behaviour of \lpnormmwc exhibits a sharp transition from $p=1$ to $p>1$. Our work raises several open questions. We mention a couple of them: (1) Can we achieve an $O(\log{n})$-approximation for \lpnormmwc for every $p\ge 1$? We recall that when $p=\infty$, the current best approximation factor is indeed $O(\log{n})$ \cite{BFKMNNS14}. (2) Is there a polynomial-time  algorithm for \lpnormmwc for any given $p$ that achieves an approximation factor that smoothly interpolates between the best possible approximation for $p=1$ and the best possible approximation for $p=\infty$---e.g., is there an $O(\log^{1-1/p}{n})$-approximation?

%\knote{Is \lpnormmwc solvable in poly-time for planar graphs when $k$ is fixed? DJPSY94 show that \mwc is poly-time for planar graphs when $k$ is fixed. Svitkina-Tardos show that \mmmwc is only $(2+\epsilon)$-approximable in trees---what about planar graphs? See BFKMNNS.}

%\knote{What about minor-free graphs? Can we get better approximations in those? See BFKMNNS.}

\bibliographystyle{amsplain}
\bibliography{references}

%\section*{Appendix}
%\appendix
%In this appendix, we present all proofs that were missing in the main body due to space limitations.

\end{document}